\documentclass[journal,twoside,web]{ieeecolor}

\usepackage{generic}
\usepackage{cite}
\usepackage{amsmath,amssymb,amsfonts}
\usepackage{algorithmic}
\usepackage{graphicx}
\usepackage{algorithm,algorithmic}
\usepackage{hyperref}
\hypersetup{hidelinks}
\usepackage{textcomp}
\usepackage{empheq}
\usepackage{microtype}

\usepackage{bm}

\def\BibTeX{{\rm B\kern-.05em{\sc i\kern-.025em b}\kern-.08em
    T\kern-.1667em\lower.7ex\hbox{E}\kern-.125emX}}
\newtheorem{thm}{\it Theorem}
\newtheorem{lem}{\it Lemma}
\newtheorem{asmp}{\it Assumption}
\newtheorem{problem}{\it Problem}
\newtheorem{remark}{\it Remark}
\newtheorem{corol}{\it Corollary}
\newtheorem{prop}{\it Proposition}
\newtheorem{define}{\it Definition}

\begin{document}
\title{Approximate Simulation-based Hierarchical Control of Nonlinear Systems}
\author{Zirui Niu, \IEEEmembership{Graduate Student Member, IEEE}, Antoine Girard, \IEEEmembership{Fellow, IEEE}, \\ Giordano Scarciotti, \IEEEmembership{Senior Member, IEEE}, 
\thanks{Z. Niu and G. Scarciotti are with the Department of Electrical and Electronic Engineering, Imperial College London, SW7 2AZ, London, U.K. {\tt\small [zn120, gs3610]@ic.ac.uk} }
\thanks{A. Girard is with Université Paris-Saclay, CNRS, CentraleSupélec, Laboratoire des Signaux et Systèmes, 91190, Gif-sur-Yvette, France (e-mail: 
antoine.girard@centralesupelec.fr). } 
}
\maketitle

\begin{abstract}
Controlling complex dynamical systems to satisfy sophisticated specifications remains a significant challenge in modern engineering. A promising approach to this problem is the approximate simulation-based hierarchical control (ASHC) technique. In this method, a simplified representation of the complex system, called the abstract system, is first designed and controlled. An interface function is then designed to translate the control law into the input of the complex system, thereby achieving approximate control synthesis. However, most existing results in ASHC are only for linear systems. This paper proposes a constructive method for solving the ASHC problem for nonlinear systems. To this end, we propose invariance equation-based methods to achieve the two classical requirements of the ASHC technique, namely the bounded output discrepancy and the $m$-relation. We then study the solvability conditions of the problem and summarise the overall design procedures. We illustrate the results with a practical example, providing step-by-step solutions to the ASHC problem of a DC-to-DC Ćuk converter.
\end{abstract}


\section{Introduction}\label{sec:intro}
\IEEEPARstart{M}{odern} engineering applications, such as robotics, autonomous vehicles, aeronautics, and manufacturing, increasingly rely on complex systems (\textit{e.g.}, nonlinear and/or higher-order) to enable higher levels of autonomy and precision. The inherent complexity of such models, however, imposes substantial computational demands on analysis and control design, thereby incurring significant processing time and energy expenditure in simulations and hardware implementations. Aiming to handle sophisticated specifications with such complexities can even result in intractable computations. Such a problem is intrinsically multidisciplinary, and it is dealt with in different ways by different research communities. A notable framework to solve this problem is model order reduction (MOR), see, \textit{e.g.},~\cite{ref:antoulas2005approximation}  for linear systems and~\cite{ref:scarciotti2024interconnection} for nonlinear systems. Another family of approaches, which is particularly indicated when the key focus of the simplification is control design, is based on hierarchical control. This is a framework introduced by Girard and Pappas~\cite{ref:girard2009hierarchical} based on the notion of approximate simulation, and hereafter referred to as approximate simulation-based hierarchical control (ASHC). As illustrated in Fig.~\ref{fig:Hierarchical_Structure}, given a complex system to be controlled, referred to as the \textit{concrete system}, a simpler model is constructed, termed the \textit{abstract system} or \textit{abstraction}. An interface function is then designed to translate the control inputs synthesised for the abstraction into inputs applicable to the concrete system, thereby enabling approximate control of the concrete system while providing guarantees on the output error. While~\cite{ref:girard2009hierarchical} only proposed solutions for linear systems, the approach was then extended to more complicated scenarios such as networked systems~\cite{ref:zamani2017compositional,ref:rungger2016compositional,ref:smith2020approximate}, piecewise
affine systems~\cite{ref:song2022robust}, linear stochastic hybrid systems~\cite{ref:zamani2016approximations}, and cyber-physical systems~\cite{ref:zhong2024hierarchical}.

\begin{figure}[tbp]
\begin{centering}
    \includegraphics[width=0.75\linewidth]{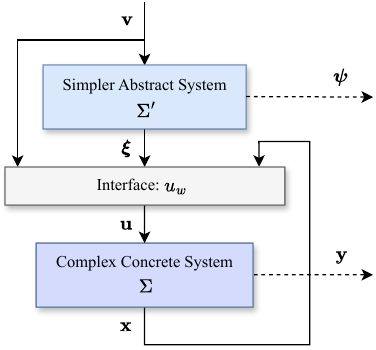}
    \caption{Hierarchical control system architecture. }
    \label{fig:Hierarchical_Structure}
\end{centering}
\end{figure}

However, there is no constructive way to solve the ASHC problem for general nonlinear systems. Solutions available in the existing literature only apply to special cases. For example, Fu et al.~\cite{ref:fu2013hierarchical} solved the ASHC problem of differentially flat nonlinear systems by transforming the nonlinear system into a linear model and proceeding with the linear method in~\cite{ref:girard2009hierarchical}. Zamani and Arcak~\cite{ref:zamani2017compositional} solved the ASHC problem for Lur'e-type systems, building upon their compositional ASHC solutions for networked systems. These existing results cannot be extended to the general nonlinear cases. It is also worth mentioning that~\cite{ref:tabuada2008approximate} proposed an approximate simulation approach that constructs a finite transition system from the concrete system for symbolic control synthesis. However, by relying on grid-based discretisation, this method suffers from the curse of dimensionality, rendering it computationally expensive for higher-dimensional systems, while also obscuring the physical intuition of the underlying continuous dynamics. Therefore, the ASHC of nonlinear systems remains an open problem.

In this paper, we solve the ASHC problem for general nonlinear systems, with a particular focus on input-affine systems. The classical ASHC problem typically requires two key design criteria, namely the bounded output discrepancy property and the so-called \textit{``$\mathit{\Pi}$-relation''}. The bounded output discrepancy requires that the error between the outputs of the concrete and abstract systems is bounded by a known function called \textit{simulation function}, originated from~\cite{ref:girard2007approximation}. Furthermore, the $\Pi$-relation requires that the construction of an abstract system preserves the input-output controllability properties of the concrete system. In this paper, we address both requirements in the nonlinear setting. First, under certain stabilisability and Lipschitz continuity assumptions, we characterise the bounded output discrepancy requirement via a partial differential equation (PDE). Then we generalise the linear $\Pi$-relation to an $m$-relation for nonlinear systems and characterise the requirement through another PDE. By combining the two PDE-based characterisations, we establish the overall solvability condition and present a detailed procedure for addressing the ASHC problem. 
Finally, we demonstrate the approach by applying the results step by step to a DC-to-DC Ćuk converter.

The main contribution of this paper is to provide constructive methods for solving ASHC problems for nonlinear systems. It is worth mentioning that, unlike MOR techniques that focus on open-loop trajectory accuracy, see, \textit{e.g.},~\cite{ref:antoulas2005approximation,ref:schilders2008model,ref:scarciotti2024interconnection} and references therein, the ASHC technique focuses on control design. However, since~\cite{ref:niu2025Briging} established a conceptual bridge between the ASHC technique and the moment-matching-based MOR technique, the ASHC technique can also be regarded as a model order reduction method for control design. From this perspective, although several studies exist, \textit{e.g.},~\cite{ref:obinata2012model,ref:meijer2022finite,ref:jansen2017use}, most of the available results are confined to linear systems with certain assumptions on system/controller structures. Therefore, the contribution of this work can be alternatively viewed as a novel approach to employing reduced-order models for controlling nonlinear systems. 


The remainder of this paper is organised as follows. Section~\ref{sec:prelimiaries} recalls some preliminary notions. Section~\ref{sec:formulate} formulates the ASHC problem by extending its two key classical requirements to the nonlinear setting. Then, the solution to the ASHC problem is provided in Section~\ref{sec:solution}, which first addresses each key requirement individually and then outlines the overall design procedures to satisfy both requirements. 
An illustrative example based on the DC-to-DC Ćuk converter, along with step-by-step solutions, is presented in Section~\ref{sec:example}, which is followed by some concluding remarks in Section~\ref{sec:concl}.

\textbf{Notation.} The symbols $\mathbb{R}$, $\mathbb{R}_{\geq 0}$, and $\mathbb{R}_{> 0}$ denote the set of real numbers, non-negative real numbers and positive real numbers, respectively. The symbol $I_{n}$ indicates an $n$-dimensional identity matrix and the symbol $\bm{0}_{m \times n}$ represents a zero matrix in $\mathbb{R}^{m \times n}$. 
The superscript $\top$ denotes the transposition operator and the superscript $\dagger$ denotes the Moore-Penrose pseudoinverse operator. Given any vector $x \in \mathbb{R}^n$ with a positive definite matrix $Q \in \mathbb{R}^{n \times n}$, $\|x\|$ indicates its Euclidean norm and $\|x\|_Q$ represents its weighted norm $\sqrt{x^\top Q x}$. Given a bounded function $\mathbf{x}: \mathbb{R} \rightarrow \mathbb{R}^{n}$, $\|\mathbf{x}\|_{\infty} \coloneqq \sup _{t \geq 0}\|\mathbf{x}(t)\|$ denotes its $\mathcal{L}^{\infty}$ norm. A function $\gamma: \mathbb{R}_{\geq 0} \rightarrow \mathbb{R}_{\geq 0}$ belongs to class $\mathcal{K}$ if it is continuous, strictly increasing, and satisfies $\gamma(0)=0$. A function $\gamma$ is said to belong to class $\mathcal{K}_\infty$ if $\gamma \in \mathcal{K}$ and $\gamma(r) \to +\infty$ as $r \to +\infty$. A continuous function $\beta: \mathbb{R}_{\geq 0} \times \mathbb{R}_{\geq 0} \rightarrow \mathbb{R}_{\geq 0}$ belongs to class $\mathcal{KL}$ if, for each fixed $t$, the function $\beta(\cdot, t)$ belongs to class $\mathcal{K}$ and, for each fixed $r$, $\beta(r, \cdot)$ is strictly decreasing and satisfies $\lim_{t \to +\infty}\beta(r, t) = 0$. Given any function $f: \mathbf{N} \to \mathbf{M}$ for some set $\mathbf{N} \subseteq \mathbb{R}^n$ and $\mathbf{M} \subseteq \mathbb{R}^m$, the image (also called range) of $f$ is defined as $\operatorname{Im}(f) := \{f(x) \, | \, x \in \mathbf{N}\} \subseteq \mathbf{M}$, while the kernel of $f$ is denoted by $\operatorname{Ker}(f) := \{x \in \mathbf{N} \,|\, f(x) = 0 \} \subseteq \mathbf{N}$.

\section{Preliminaries}\label{sec:prelimiaries}
In this section, we formulate the nonlinear systems studied in this paper, and we recall the definition of stabilisability for nonlinear systems.

\subsection{Nonlinear Control Systems}
The class of continuous-time nonlinear control systems to be studied is defined as follows.
\begin{define}\label{def:ControlSystem}
    A control system $\Sigma$ is a tuple $\Sigma= \left(\mathbb{R}^n, \mathbf{U}, \mathcal{U}, f, \mathbb{R}^{p}, h\right)$, where $\mathbb{R}^n, \mathbb{R}^p$ are the state and output spaces, respectively, and
    \begin{itemize}
        \item $\mathbf{U} \subseteq \mathbb{R}^{m}$ is the input space containing the origin.
        \item $\mathcal{U}$ is a set of all measurable functions of time, from some open interval $\ell \subseteq \mathbb{R}$ to $\mathbf{U}$.
        \item $f: \mathbb{R}^n \times \mathbb{R}^m \rightarrow \mathbb{R}^n$ is a locally Lipschitz continuous mapping, \textit{i.e.}, given any compact set $\mathbf{X}_l \subset \mathbb{R}^n$, there exists a constant $L_f \in \mathbb{R}_{> 0}$ such that $\|f(x, u)-f(z, u)\| \leq L_f\|x-z\|$ for all $x, z \in \mathbf{X}_l$. 
        \item $h: \mathbb{R}^n \rightarrow \mathbb{R}^{p}$ is the output mapping.\hfill$\blacksquare$
    \end{itemize}
\end{define}

Then for any open interval $\ell \subseteq \mathbb{R}$, we say that a locally absolutely continuous function $\mathbf{x}(t) \in \mathbb{R}^n$ and a curve $\mathbf{y}(t) \in \mathbb{R}^p$ are the \textit{state trajectory} and the \textit{output trajectory} of system~$\Sigma$ if there exists an input function $\mathbf{u} \in \mathcal{U}$ satisfying
\begin{equation}\label{equ:ConcreteSystem}
    \Sigma: 
    \left\{\begin{array}{rl}
        \dot{\mathbf{x}} \!\!\!\!&= f(\mathbf{x}, \mathbf{u}),   \\
         \mathbf{y} \!\!\!\!&= h(\mathbf{x}), 
    \end{array}\right.
\end{equation}
for almost all $t \in \ell$. Throughout this paper, we assume that all considered control systems of the form~$\Sigma$ are forward complete~\cite{ref:angeli1999forward}, \textit{i.e.},
given any state trajectory with initial condition $\mathbf{x}(t_0) = x_0 \in \mathbb{R}^n$ and controlled by any input $\mathbf{u} \in \mathcal{U}$, there exists a function $\Phi_f^x: \mathbb{R} \times \mathbb{R} \times \mathbb{R}^n \times \mathbb{R}^m \rightarrow \mathbb{R}^n$ such that
$$
\mathbf{x}(t)=\Phi_f^x\left(t, t_0, x_0, \mathbf{u}_{[t_0, t)}\right)
$$
for all $t \in [t_0, +\infty)$. The uniqueness of $\Phi_f^x(t_f, t_0, x_0, \mathbf{u})$ is guaranteed by the Lipschitz property of $f$, see~\cite{ref:sontag2013mathematical}.


From now on, we indicate with~$\Sigma = \left(\mathbb{R}^n, \mathbf{U}, \mathcal{U}, f, \mathbb{R}^{p}, h\right)$ the concrete system, \textit{i.e.}, the complex nonlinear system that we want to control. We indicate with~$\Sigma^\prime = \left(\mathbb{R}^{\hat{n}}, \mathbb{R}^{\hat{m}}, \mathcal{V}, \phi, \mathbb{R}^{p}, \kappa\right)$ the (simpler) abstract system, which has state and output trajectories satisfying
\begin{equation}\label{equ:AbstractSystem}
\Sigma^{\prime}:
\left\{\begin{array}{rl}
\dot{\bm{\xi}}\!\!\!\!& = \phi(\bm{\xi}, \mathbf{v}), \\
\bm{\psi}\!\!\!\!& = \kappa(\bm{\xi}),
\end{array}\right.
\end{equation}
for some input $\mathbf{v} \in \mathcal{V}$. Note that both systems~$\Sigma$ and $\Sigma^{\prime}$ have the same observation space (\textit{i.e.}, $\mathbb{R}^{p}$), but may possess different state and input spaces. 


\subsection{Notion of Stabilisability}
We now recall some stabilisability conditions that will be used later in this paper. For brevity, denote $\Delta := \left\{(x, z) \in \mathbb{R}^n \times \mathbb{R}^n \mid x=z\right\}$ as a diagonal set on $\mathbb{R}^{2n}$. Then we have the following definition.

\begin{define}[Asymptotic Stabilisability~\cite{ref:tabuada2008approximate}]\label{def:stabilisab}
A control system~$\Sigma = \left(\mathbb{R}^n, \mathbf{U}, \mathcal{U}, f, \mathbb{R}^{p}, h\right)$ is said to be \textit{asymptotically stabilisable\footnote{The property was originally called Stabilizability Assumption II in~\cite[Definition 2.5]{ref:tabuada2008approximate}.}} if there exists a function $k: \mathbb{R}^n \times \mathbb{R}^n \times \mathbf{U} \rightarrow \mathbf{U}$ satisfying:
\begin{itemize}
    \item [i)] $k$ is continuously differentiable on $\mathbb{R}^{2 n} \backslash \Delta$;
    \item [ii)] $k(z, x, u)=u$ for all $(x, z) \in \Delta$;
\end{itemize}
and rendering the system $(\dot{\mathbf{x}}, \dot{\mathbf{z}}) \!=\! \left(f \times_k f\right)((\mathbf{x}, \mathbf{z}), \mathbf{u})$ where
\begin{equation}\label{equ:funcff}
   \left(f \times_k f\right)((x, z), u):=(f(x, u), f(z, k(z, x, u))) 
\end{equation}
uniformly globally asymptotically stable with respect to $\Delta$, that is, enforcing that for all $x_0, z_0 \in \mathbb{R}^n$, $\ell = [t_0, +\infty)$ with $t_0 \in \mathbb{R}$, $\mathbf{u} \in \mathcal{U}$, and $t \in \ell$, there exists a class $\mathcal{K}\mathcal{L}$ function $\beta_k$ satisfying
\begin{equation}\label{equ:trajConverge}
    \|\Phi_f^x(t, t_0, x_0, \mathbf{u}_{[t_0, t)}\!)-\Phi_f^z(t, t_0, z_0, \mathbf{u}_{[t_0, t)}^*\!)\| \!\leq\! \beta_k(\|x_0 - z_0\|, t),
\end{equation}
where the input function $\mathbf{u}^*(t) := k(\mathbf{z}(t), \mathbf{x}(t), \mathbf{u}(t)) = k(\Phi_f^z(t, t_0, z_0, \mathbf{u}_{[t_0, t)}^*), \Phi_f^x(t, t_0, x_0, \mathbf{u}_{[t_0, t)}), \mathbf{u}(t))$.
\hfill$\blacksquare$
\end{define}

By Definition~\ref{def:stabilisab}, a control system~$\Sigma$ is asymptotic stabilisable if for any state trajectory $\mathbf{x}$ controlled by an input $\mathbf{u}$, there exists a feedback control law $\mathbf{u}^* = k(\mathbf{z}, \mathbf{x}, \mathbf{u})$ such that another state trajectory $\mathbf{z}$ starting with any initial condition $z_0$ asymptotically converges to $\mathbf{x}$. This property can be checked by a Lyapunov criterion as follows.

\begin{thm}\label{thm:StabiliseLyap}\cite[Proposition 2.7]{ref:tabuada2008approximate}
A control system $\Sigma=\left(\mathbb{R}^n, \mathbf{U}, \mathcal{U}, f, \mathbb{R}^{p}, h\right)$ is asymptotically stabilisable if and only if there exist:
\begin{itemize}
    \item[i)] a function $k: \mathbb{R}^n \times \mathbb{R}^n \times \mathbf{U} \rightarrow \mathbf{U}$ satisfying:
    \begin{itemize}
    \item[a)] $k$ is continuously differentiable on $\mathbb{R}^{2 n} \backslash \Delta$;
    \item[b)] $k(z, x, u)=u$ for all $(x, z) \in \Delta$,
    \item[c)] $f \times_k f$ defined in~(\ref{equ:funcff}) is forward complete.
    \end{itemize}
    \item[ii)] a smooth function $V: \mathbb{R}^n \times \mathbb{R}^n \rightarrow \mathbb{R}_{\geq 0}$ and class $\mathcal{K}_{\infty}$ functions $\underline{\alpha}, \bar{\alpha}, \alpha$ for which the following inequalities hold for all $x, z \in \mathbb{R}^n$ and $u \in \mathbf{U}$:
    \begin{subequations}\label{equ:stabilisableV}
        \begin{align}
            \!\!\underline{\alpha}(\|x-z\|) &\leq V(x, z) \leq \bar{\alpha}(\|x-z\|), \label{equ:stabilisableCond1} \\
            \!\!\frac{\partial V}{\partial x} f(x, u) + \frac{\partial V}{\partial z}& f(z, k(z, x, u)) \leq-\alpha(\|x-z\|). \label{equ:stabilisableCond2}
        \end{align}
    \end{subequations}
\end{itemize}
\end{thm}
\hfill$\blacksquare$

Instead of checking~(\ref{equ:trajConverge}), the Lyapunov criterion in Theorem~\ref{thm:StabiliseLyap} provides a constructive way of characterising the asymptotic stabilisability property. Note also that, by comparing this property with the incremental stability in~\cite{ref:angeli2002lyapunov}, any incrementally stable system is asymptotically stabilisable with $k(z, x, u) = u$ for all $(x, z) \in \mathbb{R}^{2n}$. In other words, asymptotic stabilisability is a less restrictive condition than incremental stability.

\section{ASHC Problem: Formulation}\label{sec:formulate}
In this section, we formulate the ASHC problem for the class of nonlinear systems~$\Sigma = \left(\mathbb{R}^n, \mathbf{U}, \mathcal{U}, f, \mathbb{R}^{p}, h\right)$. To this end, we first need to introduce some definitions and initial results that allow us to extend the two key design requirements of the classical ASHC problem, namely the bounded output discrepancy and the $m$-relation.


\subsection{Requirement 1: Bounded Output Discrepancy}\label{subsec:BODintro}
To provide guarantees on the trajectories of the concrete system when implementing the ASHC technique shown in Fig.~\ref{fig:Hierarchical_Structure}, one key requirement is that the output discrepancy of the two systems~$\Sigma$ and $\Sigma^{\prime}$ is bounded through a so-called ``simulation function'', which was originally introduced by~\cite{ref:girard2007approximation} as a generalisation of the concept of simulation relation, see~\cite{ref:milner1989communication,ref:pappas2000hierarchically,ref:tabuada2005hierarchical}. Now we define the simulation function introduced in~\cite{ref:girard2009hierarchical} with slight modifications.

\begin{define}[Simulation Function]\label{def:simuFunction}
Consider systems~$\Sigma$ and~$\Sigma^{\prime}$. Given a set $\mathbf{V} \subseteq \mathbb{R}^{\hat{n}}$, a continuously differentiable function $W: \mathbf{V} \times \mathbb{R}^n \rightarrow \mathbb{R}_{\geq 0}$ is called a \textit{local simulation function} of system~$\Sigma^{\prime}$ to system~$\Sigma$ and a function $u_{w} : \mathbf{V} \times \mathbb{R}^n  \times \mathbb{R}^{\hat{m}} \rightarrow \mathbf{U}$ is an associated \textit{local interface} if there exist functions $\alpha_h$, $\eta \in \mathcal{K}_\infty$ and $\gamma \in \mathcal{K}$ such that 
\begin{equation}\label{equ:SimuFuncCondi1}
    W(\xi, x) \geq \alpha_h\left(\|\kappa(\xi) - h(x)\|\right),
\end{equation}
and
\begin{equation}\label{equ:SimuFuncCondi2}
    \nabla W(\xi, x)^T\left[\begin{array}{c}\phi(\xi, v) \\
    f(x, u_w(\xi, x, v))
\end{array}\right] < -\eta\left(W(\xi, x)\right) + \gamma(\|v\|),
\end{equation}
for all $(\xi, x) \in \mathbf{V} \times \mathbb{R}^n$ and $v \in \mathbb{R}^{\hat{m}}$. 
\hfill$\blacksquare$
\end{define}

There are two main differences between Definition~\ref{def:simuFunction} and~\cite[Definition 1]{ref:girard2009hierarchical}. On the one hand, Definition~\ref{def:simuFunction} is a local version of~\cite[Definition 1]{ref:girard2009hierarchical} where $\mathbf{V} = \mathbb{R}^{\hat{n}}$. The reason for this change will be explained later. On the other hand, the functions $\alpha_h$ and $\eta$ in~\cite[Definition 1]{ref:girard2009hierarchical} are assumed to be identity mappings. Despite these changes, similarly to~\cite[Theorem 1]{ref:girard2009hierarchical}, the simulation function defined in Definition~\ref{def:simuFunction} still quantifies a bound on the error of the outputs of systems~$\Sigma$ and~$\Sigma^{\prime}$, as shown in the next result.

\begin{thm}\label{thm:ErrorBound}
    Consider the hierarchical structure in Fig.~\ref{fig:Hierarchical_Structure}. Suppose $W$ is a local simulation function from~$\Sigma^{\prime}$ to $\Sigma$ on some set $\mathbf{V} \subseteq \mathbb{R}^{\hat{n}}$ and $u_w$ is the associated local interface. Then, there exists a function $\beta \in \mathcal{K}\mathcal{L}$ such that for any measurable input function $\mathbf{v}: [t_0, +\infty) \to \mathbb{R}^{\hat{m}}$ enforcing the state trajectory $\bm{\xi}(t) \in V$ for all $t \geq t_0$, the output trajectories $\bm{\psi}$ of system~$\Sigma^\prime$ and $\mathbf{y}$ of system~$\Sigma$ satisfy
    \begin{equation}\label{equ:ErrorBound}
    \begin{aligned}
        \|\bm{\psi}(t) - \mathbf{y}(t)\| \leq  \alpha_h^{-1} & (\beta(W(\xi_0, x_0), t))\\ &\!\!\!\!\!\!+ \alpha_h^{-1}(\eta^{-1}(2\gamma(\|\mathbf{v}\|_{\infty})))
    \end{aligned}
    \end{equation}
    for all $t \geq t_0$ with any initial conditions $\bm{\xi}(t_0) = \xi_0 \in \mathbf{V}$ and $\mathbf{x}(t_0) = x_0 \in \mathbb{R}^n$ at any initial time $t_0 \in \mathbb{R}$.
\end{thm}
\begin{proof}
    For any $t \geq t_0$, the satisfaction of the inequality~(\ref{equ:SimuFuncCondi2}) implies that 
    \begin{equation*}
    \begin{aligned}
        W(\bm{\xi}(t), \mathbf{x}(t)) &= W(\xi_0, x_0) + \int_{t_0}^{t} \frac{d W(\bm{\xi}(\tau), \mathbf{x}(\tau))}{d\tau} d\tau \\
        & \leq \int_{t_0}^{t} \left(-\eta\left(W(\bm{\xi}(\tau), \mathbf{x}(\tau))\right) + \gamma(\|\mathbf{v}\|_\infty)\right) d\tau.
    \end{aligned}
    \end{equation*}
    Then~\cite[Lemma 3.6]{ref:zamani2016approximations} yields the existence of a $\mathcal{K}\mathcal{L}$ function $\beta$ such that
    \begin{equation*}
        W(\bm{\xi}(t), \mathbf{x}(t)) \leq \max \left\{ \beta(W(\xi_0, x_0), t), \eta^{-1}(2\gamma(\|\mathbf{v}\|_\infty)) \right\},
    \end{equation*}
    which, by the inequality~(\ref{equ:SimuFuncCondi1}), yields that
    \begin{equation*}
    \begin{aligned}
        \|\bm{\psi}&(t) - \mathbf{y}(t)\| \leq \alpha_h^{-1}(W(\bm{\xi}(t), \mathbf{x}(t)))\\
        &\leq \max \left\{ \alpha_h^{-1}(\beta(W(\xi_0, x_0), t)), \alpha_h^{-1}(\eta^{-1}(2\gamma(\|\mathbf{v}\|_\infty))) \right\}.
    \end{aligned}
    \end{equation*}
     This result implies the satisfaction of the inequality~(\ref{equ:ErrorBound}).
\end{proof}

Theorem~\ref{thm:ErrorBound} provides a guarantee on the discrepancy between the outputs of systems~$\Sigma$ and~$\Sigma^\prime$. This guarantee is fundamental when implementing the ASHC technique. For example, consider the task of controlling a concrete system~$\Sigma$ such that the output $\mathbf{y}(t)$ tracks a given reference trajectory $\mathbf{r}: [t_0, +\infty) \to \mathbf{Y} \subseteq \mathbb{R}^{p}$ under the constraint $\|\mathbf{r}(t) - \mathbf{y}(t)\| \leq \epsilon$ for some $\epsilon \geq 0$ and $t_0 \in \mathbb{R}$. While we assume that it is possible to control $\mathbf{y}(t) = \mathbf{r}(t)$ for all $t \geq t_0$, directly controlling system~$\Sigma$ can result in high complexity in controller design and implementation. Alternatively, one can first design a controller for an abstract system~$\Sigma^{\prime}$ with $\bm{\psi}(t) = \mathbf{r}(t)$ with some input trajectory $\mathbf{v}(t)$. Then by implementing ASHC onto system~$\Sigma$, the error constraint $\|\mathbf{r}(t) - \mathbf{y}(t)\| \leq \epsilon$ is guaranteed to be satisfied if $\mathbf{v}(t)$ is such that 
\begin{equation}\label{equ:TrackConst}
    \alpha_h^{-1}(\beta(W(\xi_0, x_0), t)) + \alpha_h^{-1}(\eta^{-1}(2\gamma(\|\mathbf{v}\|_{\infty}))) \leq \epsilon.
\end{equation}
Note that the first term in the left side of~(\ref{equ:TrackConst}) can be minimised by some appropriate choice of the initial conditions $\xi_0$ and $x_0$.

From the above reference-tracking example, we can see that, as a control problem generally specifies an output subspace $\mathbf{Y} \subseteq \mathbb{R}^{p}$ of interest, a simulation function and an interface in Definition~\ref{def:simuFunction} need not exist for all $\xi \in \mathbb{R}^{\hat{n}}$, as the abstract system, by the ASHC technique, is normally controlled to operate within a certain state subspace $\mathbf{V}$ to generate some expected output trajectories $\bm{\psi}(t) \in \mathbf{Y}$. In such a case, a basic requirement for the set $\mathbf{V}$ is that
\begin{equation}\label{equ:SimuLocalRequire}
    \mathbf{V} \subseteq \mathbb{R}^{\hat{n}} : \forall \psi \in \mathbf{Y}, \exists \,\xi \in \mathbf{V} \text{ s.t. } \psi = \kappa(\xi),
\end{equation}
which ensures that when the abstract system~$\Sigma^\prime$ is operating with the expected output $\bm{\psi}(t) \in \mathbf{Y}$, the simulation function and the interface can guarantee the boundedness in the output discrepancy between the two systems~$\Sigma$ and~$\Sigma^\prime$. 


\begin{remark}
    The simulation function in Definition~\ref{def:simuFunction} has a similar mathematical structure to the one used in~\cite[Definition 3.2]{ref:zamani2017compositional}. However, key distinctions exist, as~\cite[Definition 3.2]{ref:zamani2017compositional} imposes a global definition and does not explicitly specify an associated interface function. Moreover, despite this structural similarity, Theorem~\ref{thm:ErrorBound} and its proof establish a tighter bound than that presented in~\cite[Theorem 3.3]{ref:zamani2017compositional}.
\end{remark}

\subsection{Requirement 2: m-Relation}

While the simulation function provides guarantees on the output discrepancy when implementing the ASHC technique, another classical requirement on the abstraction design, as discussed in~\cite{ref:girard2009hierarchical}, is the preservation of all possible output behaviours of the concrete system. To characterise this property, we introduce the following definition.

\begin{define}\label{def:mRelat}
    Given a set $\mathbf{Y} \subseteq \mathbb{R}^p$, let $\mathbf{X}_y = \{ x \in \mathbb{R}^n \,|\, h(x) \in \mathbf{Y}\}$. Suppose for any $u \in \mathbf{U}$ and $x \in \mathbf{X}_y$, there exists a mapping $m: \mathbf{X}_y \to \mathbb{R}^{\hat{n}}$ and a vector $v \in \mathbb{R}^{\hat{m}}$ such that 
    \begin{subequations}\label{equ:mRelat}
    \begin{align}
        \frac{\partial m(x)}{\partial x} f(x, u) &= \phi(m(x), v), \label{equ:mRelatDiff}\\
        h(x) &= \kappa(m(x)). \label{equ:mRelatOut}
    \end{align} 
    \end{subequations}
    Then we say that system~$\Sigma$ is \textit{$m$-related} to system~$\Sigma^\prime$ on $\mathbf{Y}$.
\end{define}
\hfill$\blacksquare$

Definition~\ref{def:mRelat} is a generalisation of the $\Pi$-relation defined in~\cite[Definition 3]{ref:girard2009hierarchical}, which guarantees, in the linear time-invariant (LTI) case, that all possible output trajectories of the concrete system can be reproduced by the abstract system. Similarly, this $m$-relation in Definition~\ref{def:mRelat} implies the following result.

\begin{thm}\label{thm:mRelatOutMatch}
    Given a set $\mathbf{Y} \subseteq \mathbb{R}^p$. If system~$\Sigma^\prime$ is $m$-related to system~$\Sigma$ on $\mathbf{Y}$, then for any output trajectory $\mathbf{y}: \ell \to \mathbf{Y}$ of system~$\Sigma$ on some open interval $\ell \subseteq \mathbb{R}$, $\bm{\psi}(t) = \mathbf{y}(t)$ is also an output trajectory of system~$\Sigma^\prime$ for all $t \in \ell$.
    \hfill$\blacksquare$
\end{thm}
\begin{proof}
    Recall that $\mathbf{X}_y = \{ x \in \mathbb{R}^n \,|\, h(x) \in \mathbf{Y}\}$. The satisfaction of~(\ref{equ:mRelatDiff}) implies the existence of an input $\mathbf{v}$ such that an invariant manifold $\mathcal{M}_s = \{(x, \xi) \in (\mathbf{X}_y \times \mathbb{R}^{\hat{n}}) \,|\, \xi = m(x) \}$ exists, \textit{i.e.}, for all state trajectory $\mathbf{x}: \ell \to \mathbf{X}_y$ of system~$\Sigma$, $\bm{\xi}(t) = m(\mathbf{x}(t))$ is a state trajectory of system~$\Sigma^\prime$ for all $t \in \ell$. Consequently, the satisfaction of~(\ref{equ:mRelatOut}) implies that for any output trajectory $\mathbf{y}: \ell \to \mathbf{Y}$ of system~$\Sigma$, $\mathbf{y}(t) = h(\mathbf{x}(t)) = \kappa(m(\mathbf{x}(t))) = \kappa(\bm{\xi}(t)) = \bm{\psi}(t)$ is also an output trajectory of system~$\Sigma^\prime$ for all $t \in \ell$.
\end{proof}


Theorem~\ref{thm:mRelatOutMatch} indicates the importance of the $m$-relation in various control tasks. For example, we recall the previous reference-tracking task described after Theorem~\ref{thm:ErrorBound}. An important premise for the constraint~(\ref{equ:TrackConst}) to be satisfied is that the reference trajectory $\mathbf{r}(t) \in \mathbf{Y}$ expected to be tracked by the concrete system~$\Sigma$ can be perfectly tracked by the output of the abstract system~$\Sigma^\prime$, \textit{i.e.}, $\bm{\psi}(t) = \mathbf{r}(t)$. This property, by Theorem~\ref{thm:mRelatOutMatch}, is ensured by the $m$-relation in Definition~\ref{def:mRelat}. Note also that, similarly to what is discussed in Section~\ref{subsec:BODintro}, the output trajectory $\mathbf{y}$ of system~$\Sigma$ is typically expected to vary within a specific region of interest $\mathbf{Y}$. Therefore, unlike the classical $\Pi$-relation that is defined globally for LTI systems~\cite[Definition 3]{ref:girard2009hierarchical}, the proposed $m$-relation is defined with respect to some specified output subspace $\mathbf{Y} \subseteq \mathbb{R}^{p}$ of interest, which is not necessarily equal to $\mathbb{R}^{p}$.

With all the required results introduced, we are ready to formulate the nonlinear ASHC problem.
\begin{problem}[Nonlinear ASHC]\label{prob:ASHC}
    Given a nonlinear concrete system~$\Sigma$ and an output subspace~$\mathbf{Y}$ of interest, design an abstract system $\Sigma^\prime$ and a (locally defined) interface function $u_w$ such that the following two tasks are solved. 
\begin{itemize}
    \item [(1)] For some $\mathbf{V}$ satisfying~(\ref{equ:SimuLocalRequire}), there exists a (locally-defined) simulation function $W$ for all $(\xi, x) \in \mathbf{V} \times \mathbb{R}^n$;
    \item [(2)] The abstract system $\Sigma^\prime$ is $m$-related to the concrete system~$\Sigma$ on~$\mathbf{Y}$.
\end{itemize}
\end{problem}


\section{ASHC Problem: Constructive Solution}\label{sec:solution}
In this section, we present constructive methods for designing the abstract system $\Sigma^\prime$ and the function $u_w$ that guarantee the satisfaction of each of the two key requirements: the bounded output discrepancy and the $m$-relation. To this end, in the rest of this paper we consider the design of an abstract system $\Sigma^\prime$ in the input affine form, \textit{i.e.},
\begin{equation}\label{equ:abstractAffine}
    \phi(\xi, v) = \bar{\phi}(\xi) + \delta(\xi)v,
\end{equation}
where $\bar{\phi}: \mathbb{R}^{\hat{n}} \to \mathbb{R}^{\hat{n}}$ and $\delta: \mathbb{R}^{\hat{n}} \to \mathbb{R}^{\hat{n} \times \hat{m}}$ satisfy the same regularity conditions as $\phi$, and therefore as $f$ of the concrete system~$\Sigma$. 

In fact, most solutions in this paper consider the case that system~$\Sigma$ is also input affine, \textit{i.e.},
\begin{equation}\label{equ:concreteAffine}
    f(x, u) = \bar{f}(x) + g(x)u,
\end{equation}
with $\bar{f}: \mathbb{R}^{n} \to \mathbb{R}^{n}$ and $g: \mathbb{R}^{n} \to \mathbb{R}^{n \times m}$ also satisfy the same regularity conditions as $f$.
Nevertheless, we will show that the core results can be extended to general nonlinear systems.

\subsection{Bounded Output Discrepancy}\label{subsec:BODdesign}
As discussed in Section~\ref{sec:formulate}, an essential design requirement of the ASHC problem is to guarantee the existence of a simulation function that bounds the outputs of systems~$\Sigma$ and~$\Sigma^\prime$ in the hiararchical framework in Fig.~\ref{fig:Hierarchical_Structure}. We now provide a method for characterising the simulation function and the associated interface, relying on the following assumption.
\begin{asmp}\label{asmp:outputLipstchiz}
    The output mapping $h$ of system~$\Sigma$ is Lipschitz continuous, \textit{i.e.}, there exists a constant $L_h \in \mathbb{R}_{> 0}$ such that $\|h(x)-h(z)\| \leq L_h\|x-z\|$ for all $x, z \in \mathbb{R}^n$. 
\end{asmp}

\begin{thm}\label{thm:SimuDesign}
    Consider system~$\Sigma$ satisfying~(\ref{equ:concreteAffine}) and system~$\Sigma^\prime$ satisfying~(\ref{equ:abstractAffine}). Suppose system~$\Sigma$ is asymptotically stabilisable and satisfies Assumption~\ref{asmp:outputLipstchiz}. Assume there exist a set $\mathbf{V} \subseteq \mathbb{R}^{\hat{n}}$ and mappings $p: \mathbf{V} \to \mathbb{R}^n$ and $l: \mathbf{V} \to \mathbf{U}$ that solve
    \begin{subequations}\label{equ:SimuDesign}
    \begin{align}
        \frac{\partial p(\xi)}{\partial \xi} \bar{\phi}(\xi) &= \bar{f}(p(\xi)) + g(p(\xi)) l(\xi), \label{equ:SimuDesignCond1}\\
        \kappa(\xi) &= h(p(\xi)), \label{equ:SimuDesignCond2}
    \end{align}
    \end{subequations}
    for all $\xi \in \mathbf{V}$. Then the function $W(\xi, x) = V(p(\xi), x)$ with $V$ satisfying~\eqref{equ:stabilisableV} is a local simulation function of~$\Sigma^\prime$ to~$\Sigma$ with the associated local interface 
    \begin{equation}\label{equ:Interface}
        u_w(\xi, x, v) = k(x, p(\xi), l(\xi)) + u^*(\xi, x, v)
    \end{equation}
    for some function $u^*: \mathbf{V} \times \mathbb{R}^n \times \mathbb{R}^{\hat{m}} \to \mathbf{U}$ if the function $V$ in~\eqref{equ:stabilisableV} satisfies the following boundedness condition.
    \begin{itemize}
        \item[\textbf{(B)}] There exist functions $\alpha^* \in \mathcal{K}_\infty$ and $\gamma \in \mathcal{K}$ such that the function $\alpha_d := \alpha - \alpha^* \in \mathcal{K}_\infty$ with $\alpha$ in~\eqref{equ:stabilisableCond2}, and
    \begin{equation}\label{equ:SimuResBound}
    \begin{aligned}
        \bigg\| \frac{\partial V(p(\xi), x)}{\partial p(\xi)} \frac{\partial p(\xi)}{\partial \xi} &\delta(\xi)v + \frac{\partial V(p(\xi), x)}{\partial x} g(x) u^*\bigg\| \\ 
        &\; \leq \alpha^*(\|p(\xi) - x\|) + \gamma(\|v\|)
    \end{aligned}
    \end{equation}
    holds for all $(\xi, x) \in \mathbf{V} \times \mathbb{R}^n$ and $v \in \mathbb{R}^{\hat{m}}$.
    \end{itemize}
    \vspace{0.5mm}
\end{thm}
\begin{proof}
    To prove the theorem, we need to show that the function $W(\xi, x) = V(p(\xi), x)$ with the interface $u_w$ in~\eqref{equ:Interface} satisfies~\eqref{equ:SimuFuncCondi1} and~\eqref{equ:SimuFuncCondi2}. In fact, since system~$\Sigma$ is asymptotically stabilisable and therefore $V$ satisfies~\eqref{equ:stabilisableCond1}, $W(\xi, x) = V(p(\xi), x)$ satisfies
    \begin{equation*}
        W(\xi, x) \geq \underline{\alpha}(\|p(\xi)-x\|) \geq \underline{\alpha}\left(\frac{1}{L_h}\|h(p(\xi))-h(x)\|\right),
    \end{equation*}
    where the last inequality is derived based on Assumption~\ref{asmp:outputLipstchiz}. Then, by~\eqref{equ:SimuDesignCond2}, we obtain $W(\xi, x) \geq \underline{\alpha}\bigl(\frac{1}{L_h}\|\kappa(\xi)-h(x)\|\bigr)$, implying that $W(\xi, x)$ satisfies~\eqref{equ:SimuFuncCondi1} with $\alpha_h(\cdot) = \underline{\alpha}(\frac{1}{L_h}(\cdot))$. Then to show~\eqref{equ:SimuFuncCondi2}, as functions $f$ and $\phi$ satisfy~\eqref{equ:concreteAffine} and~\eqref{equ:abstractAffine} respectively, we have
    \begin{equation*}
    \begin{aligned}
    &\nabla W(\xi, x)^T\left[\begin{array}{c}\phi(\xi, v) \\
        f(x, u_w(\xi, x, v)) 
    \end{array}\right] \\
    &= \frac{\partial W}{\partial p(\xi)} \frac{\partial p(\xi)}{\partial \xi}\left(\bar{\phi}(\xi) + \delta(\xi)v\right) + \frac{\partial W}{\partial x} \!\left(\bar{f}(x) \!+\! g(x)u_w\right) \\
    &= \frac{\partial W}{\partial p(\xi)} f(p(\xi), l(\xi)) \!+\! \frac{\partial W}{\partial x} \left( \bar{f}(x) \!+\! g(x)u_w\right) \!+\! \frac{\partial W}{\partial \xi}\delta(\xi)v,
    \end{aligned}
\end{equation*}
where $\partial W(\xi, x) /\partial p(\xi) = V(p(\xi), x) /\partial p(\xi)$, and the last equation is derived from~\eqref{equ:SimuDesignCond1}. Remember that $W(\xi, x) = V(p(\xi), x)$ with $V$ satisfying~\eqref{equ:stabilisableCond2} and the interface $u_w$ is designed as~\eqref{equ:Interface}. If condition \textbf{(B)} holds, we obtain
    \begin{equation*}
    \begin{aligned}
        &\nabla W(\xi, x)^T\left[\begin{array}{c}\phi(\xi, v) \\
        f(x, u_w(\xi, x, v)) 
    \end{array}\right] \\
    &= \frac{\partial W}{\partial p(\xi)} f(p(\xi), l(\xi)) + \frac{\partial W}{\partial x} \left( \bar{f}(x) + g(x)k(x, p(\xi), l(\xi))\right) \\
    & \quad + \frac{\partial W}{\partial p(\xi)} \frac{\partial p(\xi)}{\partial \xi} \delta(\xi)v + \frac{\partial W}{\partial x} g(x) u^* \\
    &\leq -\alpha \left(\|p(\xi) - x\|\right) + \left\|\frac{\partial W}{\partial p(\xi)} \frac{\partial p(\xi)}{\partial \xi}\delta(\xi)v + \frac{\partial W}{\partial x} g(x) u^*\right\| \\
    &\leq -\alpha \left(\|p(\xi) - x\|\right) + \alpha^*\left(\|p(\xi) - x\|\right) + \gamma(\|v\|) \qquad \qquad \qquad \\
    &\leq -\alpha_d \left(\|p(\xi) - x\|\right) + \gamma(\|v\|).
    \end{aligned}
\end{equation*}
Now consider again that $W(\xi, x) = V(p(\xi), x)$ with $V$ satisfying~\eqref{equ:stabilisableCond1}, we have $W(\xi, x) \leq \overline{\alpha}(\|p(\xi)-x\|)$, or equivalently, $\overline{\alpha}^{-1}(W(\xi, x)) \leq \|p(\xi)-x\|$. This finally implies
\begin{equation*}
    -\alpha_d \!\left(\|p(\xi) - x\|\right) + \gamma(\|v\|) \!\leq\! -\alpha_d \circ \overline{\alpha}^{-1}(W(\xi, x)) + \gamma(\|v\|),
\end{equation*}
showing that $W$ satisfies~\eqref{equ:SimuFuncCondi2} with $\eta = \alpha_d \circ \overline{\alpha}^{-1}$.
\end{proof}

Theorem~\ref{thm:SimuDesign} presents a constructive way of designing an abstract system~$\Sigma^\prime$ and an interface $u_w$ to guarantee the existence of a simulation function under the extra condition~\textbf{(B)}. We now provide a sufficient condition that can ensure the satisfaction of this condition.
\begin{prop}\label{prop:condB}
    Consider condition~\textbf{(B)}. Suppose the set $\mathbf{V} \subseteq \mathbb{R}^{\hat{n}}$ is compact, the mapping $p$ that solves~\eqref{equ:SimuDesign} is continuously differentiable, and the mapping $\delta$ is continuous for all $\xi \in \mathbf{V}$. Then there exist functions $\alpha^* \in \mathcal{K}_\infty$ and $\gamma \in \mathcal{K}$ such that~\eqref{equ:SimuResBound} holds with $u^*\equiv 0$ if there exists a function $\hat{\alpha}^* \in \mathcal{K}_\infty$ such that $V$ in~\eqref{equ:stabilisableV} satisfies
    \begin{equation}\label{equ:bounddVdx}
        \left\|\frac{\partial V(x, z)}{\partial x} \right\|^2 \leq \hat{\alpha}^*(\|x - z\|)
    \end{equation}
    for all $x, z \in \mathbb{R}^n$.
\end{prop}
\begin{proof}
Consider~\eqref{equ:SimuResBound} and let $u^* \equiv 0$. By the Cauchy-Schwarz inequality and Young’s inequality, the left-hand side of~\eqref{equ:SimuResBound} comes down to
\begin{equation}\label{equ:SimuResBoundYoung}
\begin{aligned}
    &\left\|\frac{\partial V(p(\xi), x)}{\partial p(\xi)} \frac{\partial p(\xi)}{\partial \xi}\delta(\xi)v\right\| \leq \left\|\frac{\partial V(p(\xi), x)}{\partial p(\xi)}\right\| \left\|\frac{\partial p(\xi)}{\partial \xi}\delta(\xi)v\right\| \\
    &\leq \frac{\varepsilon}{2} \left\|\frac{\partial V(p(\xi), x)}{\partial p(\xi)} \right\|^2 \!+ \frac{1}{2\varepsilon}\left\| \frac{\partial p(\xi)}{\partial \xi}\delta(\xi)\right\|^2 \|v\|^2
\end{aligned}
\end{equation}
for any $\varepsilon > 0$. Note that since $p$ is continuously differentiable with $\delta$ continuous, there exists a positive constant $\bar{d}$ such that $$\left\|\frac{\partial p(\xi)}{\partial \xi}\delta(\xi)\right\|^2 \leq \bar{d}$$
for all $\xi \in \mathbf{V}$ with $\mathbf{V}$ a compact set. Then by~\eqref{equ:bounddVdx}, there exist functions $\alpha^* = \frac{\varepsilon}{2}\hat{\alpha}^* \in \mathcal{K}_{\infty}$ and $\gamma(r) = \frac{\bar{d}}{2\varepsilon}r^2$ such that 
\begin{equation}
\begin{aligned}
    \frac{\varepsilon}{2} \left\|\frac{\partial V(p(\xi), x)}{\partial p(\xi)} \right\|^2 \!\!+&\, \frac{1}{2\varepsilon}\left\| \frac{\partial p(\xi)}{\partial \xi}\delta(\xi)\right\|^2 \|v\|^2 
    \\ 
    &\leq \alpha^*(\|p(\xi) - x\|) + \gamma(\|v\|),
\end{aligned}
\end{equation}
\textit{i.e.},~\eqref{equ:SimuResBound} holds. 
\end{proof}

By Theorem~\ref{thm:SimuDesign} and Proposition~\ref{prop:condB}, condition~\textbf{(B)} that guarantees the existence of a simulation function can be satisfied if the Lyapunov function $V$ in~\eqref{equ:stabilisableV} satisfies~\eqref{equ:bounddVdx} with $\alpha_d = \alpha - \alpha^* = \alpha - \frac{\varepsilon}{2}\hat{\alpha}^* \in \mathcal{K}_\infty$. In fact, condition~\eqref{equ:bounddVdx} can be easily satisfied if $V$ is quadratic, \textit{i.e.}, $V(x, z) = (x - z)^\top P (x - z)$ for some positive definite matrix $P \in \mathbb{R}^{n \times n}$. This is true for many incrementally stable systems, see~\cite{ref:jouffroy2010tutorial}. In this case, the requirement $\alpha_d = \alpha - \frac{\varepsilon}{2}\hat{\alpha}^* \in \mathcal{K}_\infty$ can also be naturally satisfied by an appropriate choice of the free parameter $\varepsilon$, as both $\alpha$ and $\hat{\alpha}^*$ can be quadratic.

\begin{remark}
Although Proposition~\ref{prop:condB} ensures the satisfaction of condition~\textbf{(B)} with $u^* \equiv 0$, the input function $u^*$ could be used to minimise the bound in~\eqref{equ:SimuResBound}, thereby minimising the bound~\eqref{equ:ErrorBound}, as long as condition \textbf{(B)} is satisfied. In particular, we can adopt a linear structure with respect to $v$ as
\begin{equation}\label{equ:uMinBound}
    u^*(\xi, x, v) = q(\xi, x)v
\end{equation}
for some function $q: \mathbf{V} \times \mathbb{R}^n \to \mathbb{R}^{m \times \hat{m}}$. Substituting this into~\eqref{equ:SimuResBound} yields
\begin{equation*}
    \begin{aligned}
        &\bigg\| \frac{\partial V(p(\xi), x)}{\partial p(\xi)} \frac{\partial p(\xi)}{\partial \xi} \delta(\xi)v + \frac{\partial V(p(\xi), x)}{\partial x} g(x) u^*\bigg\| \\ 
        &\leq \bigg\| \frac{\partial V(p(\xi), x)}{\partial p(\xi)} \frac{\partial p(\xi)}{\partial \xi} \delta(\xi) + \frac{\partial V(p(\xi), x)}{\partial x} g(x) q(\xi, x)\bigg\| \|v\|.
    \end{aligned}
\end{equation*}
When $q \equiv 0$, this expression naturally reduces to the bound for~\eqref{equ:SimuResBound} established by Proposition~\ref{prop:condB}. Instead, one can design $q$ to minimise the norm factor $\big\| \frac{\partial V(p(\xi), x)}{\partial p(\xi)} \frac{\partial p(\xi)}{\partial \xi} \delta(\xi) + \frac{\partial V(p(\xi), x)}{\partial x} g(x) q(\xi, x)\big\|$.
For instance, in the linear case, $q$ can simply be chosen as a constant matrix, see~\cite[Proposition 1]{ref:girard2009hierarchical}. We further illustrate this design for nonlinear systems through the example in Section~\ref{sec:example}. 
\end{remark}

\begin{remark}\label{rmk:LinearCompare}
Several similarities to the linear method in~\cite[Theorem 2]{ref:girard2009hierarchical} can be identified. Theorem~\ref{thm:SimuDesign} could be seen as a nonlinear generalisation of~\cite[Theorem 2]{ref:girard2009hierarchical}, where the asymptotic stabilisability assumption boils down to the stabilisability of the linear concrete system, and Assumption~\ref{asmp:outputLipstchiz} trivially holds. We also note that condition \textbf{(B)} naturally holds as the Lyapunov function in~\eqref{equ:stabilisableV} for linear systems is quadratic and therefore~\eqref{equ:bounddVdx} holds. In fact, the same simplification as in~\eqref{equ:SimuResBoundYoung} has been used in the linear case to prove the boundedness in~\eqref{equ:SimuResBound}, see~\cite[Theorem 2]{ref:girard2009hierarchical} and~\cite[Equation (3.12)]{ref:samari2025model}. Finally, note that the PDE in~\eqref{equ:SimuDesignCond1} is a nonlinear extension of the Sylvester equation in~\cite[Equation~(8)]{ref:girard2009hierarchical},
and it defines an invariant manifold $\mathcal{M}_d = \{(x, \xi) \in (\mathbb{R}^n \times \mathbf{V}) \,|\, x = p(\xi) \}$
in the interconnection in Fig.~\ref{fig:Hierarchical_Structure} when $\mathbf{v} \equiv 0$. 
\end{remark}



\begin{remark}\label{rmk:SimuNotAffine}
When system~$\Sigma$ is not input-affine, Theorem~\ref{thm:SimuDesign} still works with the right side of~\eqref{equ:SimuDesignCond1} replaced by $f(p(\xi), l(\xi))$. In this case, condition \textbf{(B)} can still be satisfied using Proposition~\ref{prop:condB}. Therefore, the result proposed in this subsection can be extended to general nonlinear systems.
\end{remark}

\begin{remark}\label{rmk:solvePDE}
    The method in Theorem~\ref{thm:SimuDesign} relies on the solution of the PDE in~\eqref{equ:SimuDesignCond1}. Although an analytic solution $p$ is not easy to obtain, especially for high-dimensional concrete systems, various studies have presented numerical methods for approximating such a solution for reduced-order modelling or control purposes, see, \textit{e.g.},~\cite{ref:kalise2018polynomial,ref:kalise2020robust,ref:faedo2021approximation,ref:doebeli2024polynomial}. One special case of~\eqref{equ:SimuDesignCond1} is to design $\bar{\phi}(\xi) \equiv \bm{0}_{\hat{n} \times \hat{n}}$. In this manner, the resulting equation~\eqref{equ:SimuDesignCond1} becomes an algebraic equation $f(p(\xi), l(\xi)) = 0$, obviating the need for solving PDEs.
\end{remark}

\subsection{m-Relation}\label{subsec:mRelateDesign}
As mentioned in Section~\ref{sec:formulate}, the second design requirement of the ASHC problem is the preservation of all possible output behaviours of the concrete system. This preservation is specified by the so-called $m$-relation in Definition~\ref{def:mRelat}, which depends on the solvability of equations~\eqref{equ:mRelat} for any input $u$. Now we show that this condition can be simplified when the concrete system is input-affine.

\begin{lem}\label{lem:mRelatePDE}
	Given a set $\mathbf{Y} \subseteq \mathbb{R}^p$, let $\mathbf{X}_y = \{ x \in \mathbb{R}^n \,|\, h(x) \in \mathbf{Y}\}$. Suppose systems~$\Sigma$ and~$\Sigma^\prime$ are input affine and therefore satisfies~\eqref{equ:concreteAffine} and~\eqref{equ:abstractAffine}, respectively. Then system~$\Sigma$ is $m$-related to system~$\Sigma^\prime$ on $\mathbf{Y}$ if there exist functions $m: \mathbf{X}_y \to \mathbb{R}^{\hat{n}}$, $b: \mathbf{X}_y \to \mathbb{R}^{\hat{m}}$, and $c: \mathbf{X}_y \to \mathbb{R}^{\hat{m} \times m}$ such that
	\begin{subequations}\label{equ:mRelatSuffi}
		\begin{align}
			\frac{\partial m(x)}{\partial x} \bar{f}(x) &= \bar{\phi}(m(x)) + \delta(m(x))b(x), \label{equ:mRelatSuffiDiff}\\
			\frac{\partial m(x)}{\partial x} g(x) &= \delta(m(x)) c(x), \label{equ:mRelatSuffiIn}\\
			h(x) &= \kappa(m(x)). \label{equ:mRelatSuffiOut}
		\end{align} 
	\end{subequations}
	for all $x \in \mathbf{X}_y$.
\end{lem}
\begin{proof}
	Suppose~\eqref{equ:mRelatSuffi} holds for all $u \in \mathbf{U}$ and $x \in \mathbf{X}_y$. Consider the PDE~\eqref{equ:mRelatDiff}. Since $f$ and $\phi$ satisfy~\eqref{equ:concreteAffine} and~\eqref{equ:abstractAffine}, respectively, by substituting $v = b(x) + c(x)u$,~\eqref{equ:mRelatSuffiDiff} and~\eqref{equ:mRelatSuffiIn} yields that
	\begin{equation*}
		\begin{aligned}
			\phi(m(x), v) &= \bar{\phi}(m(x)) + \delta(m(x))(b(x) + c(x)u)\\ 
			&= \bar{\phi}(m(x)) + \delta(m(x))b(x) + \delta(m(x))c(x)u \\
			&= \frac{\partial m(x)}{\partial x} \bar{f}(x) +  \frac{\partial m(x)}{\partial x}g(x)u = \frac{\partial m(x)}{\partial x} f(x, u)
		\end{aligned}
	\end{equation*}
	for any $u \in \mathbf{U}$ and $x \in \mathbf{X}_y$. As~\eqref{equ:mRelatSuffiOut} is identical to~\eqref{equ:mRelatOut}, this concludes the proof.
\end{proof}

Lemma~\ref{lem:mRelatePDE} is a generalisation of~\cite[Theorem 5]{ref:niu2025Briging} to the nonlinear case. It proposes a sufficient condition for the $m$-relation under the restriction that the input $v$ driving system~$\Sigma^\prime$ takes the form 
\begin{equation}\label{equ:mRelateLink}
   v = b(x) + c(x) u,
\end{equation}
and the resulting condition~\eqref{equ:mRelatSuffi} only depends on the parameters of systems~$\Sigma$ and~$\Sigma^\prime$. Moreover, together with Theorem~\ref{thm:mRelatOutMatch} and its proof, the satisfaction of~\eqref{equ:mRelatSuffi} also implies that the outputs of systems~$\Sigma$ and~$\Sigma^\prime$ in the interconnection in Fig.~\ref{fig:mRelateInterconnect} match one another ($\mathbf{y} \equiv \bm{\psi}$) when both systems are initialised on the manifold $\mathcal{M}_s = \{(x, \xi) \in (\mathbf{X}_y \times \mathbb{R}^{\hat{n}}) \,|\, \xi = m(x) \}$ for any input $u$. Now we show that the sufficient condition~\eqref{equ:mRelatSuffi} is not difficult to satisfy.

\begin{figure}[tbp]
\begin{centering}
    \includegraphics[width=\linewidth]{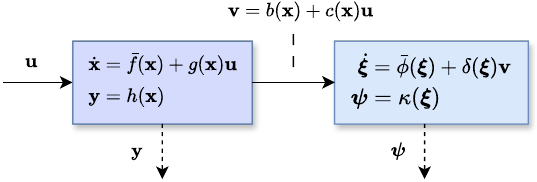}
    \caption{Interconnection to illustrate the design of $m$-relation.}
    \label{fig:mRelateInterconnect}
\end{centering}
\end{figure}

\begin{thm}\label{thm:mRelateDelta}
Consider any functions $m: \mathbf{X}_y \to \mathbb{R}^{\hat{n}}$ and $\kappa$ such that~\eqref{equ:mRelatSuffiOut} holds for all $x \in \mathbf{X}_y$. Suppose $m$ is differentiable. Then there exist functions $b: \mathbf{X}_y \to \mathbb{R}^{\hat{m}}$ and $c: \mathbf{X}_y \to \mathbb{R}^{\hat{m} \times m}$ such that~\eqref{equ:mRelatSuffiDiff} and~\eqref{equ:mRelatSuffiIn} are satisfied for all $x \in \mathbf{X}_y$ if $\delta(x)$ is of full row rank for all $x \in \mathbf{X}_y$.
\end{thm}
\begin{proof}
    Suppose $\operatorname{rank}(\delta(x)) = \hat{n}$ for all $x \in \mathbf{X}_y$. Then for any function $m$, there exist functions $b$ and $c$ solving~\eqref{equ:mRelatSuffiDiff} and~\eqref{equ:mRelatSuffiIn} for all $x \in \mathbf{X}_y$. In fact, it is possible to select
    $b(x) = \delta^\dagger(m(x)) \big(\frac{\partial m(x)}{\partial x} \bar{f}(x) - \bar{\phi}(m(x))\big)$ and $c(x) = \delta^\dagger(m(x)) \frac{\partial m(x)}{\partial x} g(x)$.
\end{proof}

Theorem~\ref{thm:mRelateDelta}, together with Lemma~\ref{lem:mRelatePDE}, yields that the $m$-relation can be sufficiently guaranteed if $\delta(x)$ is right invertible for all $x \in \mathbf{X}_y$. As a consequence, one can simply design $\delta \equiv I_{\hat{n}}$, implying the input space covers the whole state space.
Note that Theorem~\ref{thm:mRelateDelta} omitted the solvability of~\eqref{equ:mRelatSuffiOut}. If we only focus on the $m$-relation (without bounded output discrepancy requirement), this equation is easy to satisfy, as one can simply design $\kappa(\xi) = \xi$ and $m(x) = h(x)$, so long as $h$ is differentiable. However,~\eqref{equ:mRelatSuffiOut} is more difficult to satisfy if the bounded output discrepancy requirement in Section~\ref{subsec:BODdesign} must also hold. This issue is discussed in the next subsection.

\begin{figure*}[tbp]
\begin{centering}
    \includegraphics[width=\linewidth]{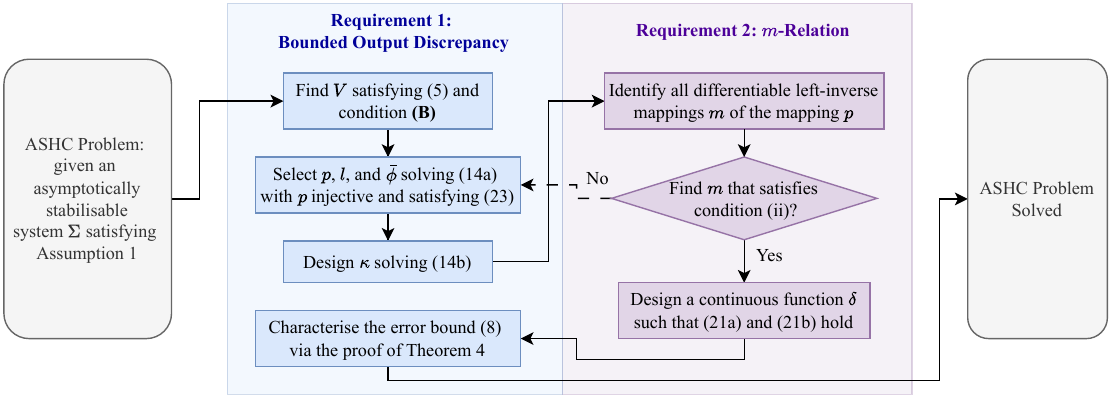}
    \caption{Flow chart of the final design procedure for solving the nonlinear ASHC problem formulated in Problem~\ref{prob:ASHC}.}
    \label{fig:Flow_Chart}
\end{centering}
\end{figure*}

\begin{remark}\label{rmk:RelateMM}
    For readers who are familiar with MOR techniques, the equations in~\eqref{equ:SimuDesign} and~\eqref{equ:mRelatSuffi} are similar to those used for nonlinear moment-matching~\cite[Theorem 7]{ref:astolfi2010model}. In fact, there exists a conceptual relation between the classical ASHC technique in~\cite{ref:girard2009hierarchical} and the moment-matching-based MOR technique~\cite{ref:astolfi2010model,ref:scarciotti2017nonlinear,ref:scarciotti2024interconnection}. This relation has been studied by~\cite{ref:niu2025Briging} in the linear context, and the results therein extend, by means of the framework presented in this paper, to the nonlinear case. For example, when comparing~\eqref{equ:SimuDesign} with the nonlinear moment matching condition~\cite[Section 3.2]{ref:scarciotti2024interconnection}, the hierarchical structure in Fig.~\ref{fig:Hierarchical_Structure} with $\mathbf{v} \equiv 0$ can be seen as a ``direct interconnection'' structure by moment matching, see~\cite[Section III-A]{ref:niu2025Briging}. In this MOR context, the abstract system~$\Sigma^\prime$ with $\delta(\xi) = 0$ behaves as a signal generator, or it can be regarded as a ``limiting'' reduced order model in the sense of moment matching, as it matches the ``moment'' (defined by $h \circ p$) of system~$\Sigma$ at $(\bar{\phi}, l)$ only when $\mathbf{v} \equiv 0$. A similar result also exists when interpreting the $m$-relation in Fig.~\ref{fig:mRelateInterconnect} through the lens of moment matching, see~\cite[Section III-B]{ref:niu2025Briging}.
\end{remark}

\subsection{Achieving Two Design Requirements}\label{subsec:steps}
Now we consider how to design the abstract system~$\Sigma^\prime$ and the interface function~$u_w$ in~\eqref{equ:Interface} to simultaneously satisfy the two aforementioned requirements in the ASHC problem, \textit{i.e.}, to solve Problem~\ref{prob:ASHC}. According to Section~\ref{subsec:BODdesign}, the bounded output discrepancy requires the determination of the mappings $p: \mathbf{V} \to \mathbb{R}^n$ and $l: \mathbf{V} \to \mathbf{U}$ that solve~\eqref{equ:SimuDesign} for some $\mathbf{V}$ satisfying~\eqref{equ:SimuLocalRequire} to guarantee the existence of a simulation function $W$. Then, according to Section~\ref{subsec:mRelateDesign}, the $m$-relation requires the satisfaction of~\eqref{equ:mRelatSuffi} for some mapping $m: \mathbf{X}_y \to \mathbb{R}^{\hat{n}}$, so that no controllable behaviour of system~$\Sigma$ on the specified state space $\mathbf{X}_y$ is disgarded when manipulating $\Sigma^\prime$. Therefore,~\eqref{equ:SimuDesign} and~\eqref{equ:mRelatSuffi} are crucial for solving the ASHC problem.

The main obstacle that hinders the concurrent solvability of~\eqref{equ:SimuDesign} and~\eqref{equ:mRelatSuffi} lies in the satisfaction of~\eqref{equ:SimuDesignCond2} and~\eqref{equ:mRelatSuffiOut}. In fact, for any functions $p$ and $l$ that solve~\eqref{equ:SimuDesignCond1}, if there exist functions $m$ and $\kappa$ such that~\eqref{equ:SimuDesignCond2} and~\eqref{equ:mRelatSuffiOut} hold (and therefore~\eqref{equ:SimuDesign} holds), then Theorem~\ref{thm:mRelateDelta} yields that~\eqref{equ:mRelatSuffiDiff} and~\eqref{equ:mRelatSuffiIn} also hold with any $\delta$ such that $\operatorname{rank}(\delta(x)) = \hat{n}$ for all $x \in \mathbf{X}_y$. The following lemma provides conditions for~\eqref{equ:SimuDesignCond2} and~\eqref{equ:mRelatSuffiOut} to hold simultaneously.

\begin{thm}\label{thm:mCondi}
    Consider any functions $p$, $l$, $\bar{\phi}$, and $\kappa$ that solve~\eqref{equ:SimuDesign}. Then~\eqref{equ:mRelatSuffiOut} holds for all $x \in \mathbf{X}_y$ if there exists a function $m$ satisfying the following conditions:
    \begin{itemize}
        \item[(i)] for any $x \in \mathbf{X}_y \cap \operatorname{Im}(p)$, we have $p(m(x)) = x$.
        \item[(ii)] for any $x \in \mathbf{X}_y \setminus \operatorname{Im}(p)$, there exists a vector $x^\prime \in \mathbf{X}_y \cap \operatorname{Im}(p)$ such that $m(x) = m(x^\prime)$ and $h(x) = h(x^\prime)$.
    \end{itemize}
\end{thm}
\begin{proof}
    Suppose $m$ satisfies condition (i). As~\eqref{equ:SimuDesignCond2} holds, we have $\kappa(m(x)) = h(p(m(x))) = h(x)$, \textit{i.e.},~\eqref{equ:mRelatSuffiOut} is satisfied for all $x \in \mathbf{X}_y \cap \operatorname{Im}(p)$. Now, suppose $m$ also satisfies condition (ii). By~\eqref{equ:SimuDesignCond2} and condition (i), we obtain $\kappa(m(x)) = h(p(m(x))) = h(p(m(x^\prime))) = h(x^\prime) =  h(x)$ for all $x \!\in\! \mathbf{X}_y \!\setminus\! \operatorname{Im}(p)$. As a result,~\eqref{equ:mRelatSuffiOut} holds for all $x \!\in\! \mathbf{X}_y$.
\end{proof}

A sufficient condition for condition~(i) in Theorem~\ref{thm:mCondi} to hold is that the function $p$ is injective. In such a case, one can define the function $m$ as a left-inverse of $p$ and therefore $m(p(\xi)) = m(x) = \xi$ holds for any $\xi \in \mathbf{V}$. Consequently, given any $x \in \mathbf{X}_y \cap \operatorname{Im}(p)$, we have $p(m(x)) = p(\xi) = x$, \textit{i.e.}, condition~(i) holds. Then $x^\prime$ in condition~(ii) satisfies $x^\prime = p(m(x))$ and such condition reduces to $h(x) = h(p(m(x)))$ for all $x \in \mathbf{X}_y \setminus \operatorname{Im}(p)$. This discussion suggests that condition~(ii) necessitates
\begin{equation}\label{equ:mRelationImP}
    h\left(\mathbf{X}_y \setminus \operatorname{Im} (p)\right) \subseteq h\left(\operatorname{Im}(p)\right).
\end{equation}
In fact, condition~(ii) can be further simplified under a certain linearity condition.

\begin{corol}\label{corol:mRelateGeo}
    Consider system~$\Sigma$. Suppose the function $h$ is linear and assume there exists a function $m$ satisfying condition (i) in Theorem~\ref{thm:mCondi}. Then condition (ii) also holds if and only if 
    \begin{equation}\label{equ:GeoCondi}
        x - p(m(x)) \in \operatorname{Ker}(h)
    \end{equation}
    for all $x \in \mathbf{X}_y \setminus \operatorname{Im}(p)$.
\end{corol}
\begin{proof}
    \textit{Sufficiency:} Suppose $p$ and $h$ satisfy~\eqref{equ:GeoCondi}. Then for any $x \in \mathbf{X}_y \setminus \operatorname{Im}(p)$, let $x^\prime = p(m(x)) \in \mathbf{X}_y \cap \operatorname{Im}(p)$ and $\bar{x} = x - p(m(x)) \in \operatorname{Ker}(h)$. Since $m$ satisfies condition~(i), we obtain $m(x^\prime) = m(p(m(x))) = m(x)$. The linearity of $h$ yields that $h(x) = h(x^\prime) + h(\bar{x}) = h(x^\prime)$, as $\bar{x} \in \operatorname{Ker}(h)$.

    \textit{Necessity:} Suppose condition (ii) holds. Then for any $x \in \mathbf{X}_y \setminus \operatorname{Im}(p)$, $m(x) = m(x^\prime)$ with $x^\prime \in \operatorname{Im}(p)$. Since $m$ also satisfies condition (i), we have $p(m(x)) = p(m(x^\prime)) = x^\prime$. Now let $\bar{x} = x - x^\prime$. As $h(x) = h(x^\prime)$ with $h$ a linear function, we obtain $h(\bar{x}) = h(x - x^\prime) = h(x) - h(x^\prime) = 0$, implying that~\eqref{equ:GeoCondi} holds.
\end{proof}

\begin{remark}
    If systems~$\Sigma$ and~$\Sigma^\prime$ are linear, the functions $p$ and $m$ become linear mappings and therefore $\operatorname{Im}(p)$ and $\operatorname{Ker}(h)$ become vector spaces. Consequently, when $\mathbf{X}_y = \mathbb{R}^n$, ~\eqref{equ:GeoCondi} boils down to the condition $\operatorname{Im}(p) + \operatorname{Ker}(h) = \mathbb{R}^n$, see~\cite[Lemma 3]{ref:girard2009hierarchical} for more details.
\end{remark}

Recall that regarding the solutions of~\eqref{equ:SimuDesign} and~\eqref{equ:mRelatSuffi}, Problem~\ref{prob:ASHC} also requires the satisfaction of~\eqref{equ:SimuLocalRequire}. This condition is, in fact, ensured by satisfying the conditions~(i) and~(ii) in Theorem~\ref{thm:mCondi} (and therefore~\eqref{equ:mRelatSuffiOut} holds) with $m$ being the left-inverse of $p$ on its domain $\mathbf{V}$. In this case, it is straightforward to see that $\operatorname{Im}(m) = \mathbf{V}$. Consequently, for any $\psi = y = h(x) \in \mathbf{Y}$ with any $x \in \mathbf{X}_y$, the satisfaction of~\eqref{equ:mRelatSuffiOut} yields the existence of $\xi = m(x) \in \mathbf{V}$ such that $\psi = y = h(x) = \kappa(\xi)$, \textit{i.e.},~\eqref{equ:SimuLocalRequire} holds.

By combining the results of Theorems~\ref{thm:mRelateDelta} and~\ref{thm:mCondi}, we summarise the steps for designing an input-affine abstract system~$\Sigma^\prime$ and an interface $u_w$ in~\eqref{equ:Interface} that guarantee~\eqref{equ:SimuDesign} and~\eqref{equ:mRelatSuffi} as follows.
\begin{enumerate}
    \item Select functions $p$, $l$, $\bar{\phi}$ solving~\eqref{equ:SimuDesignCond1} with $p$ injective such that~\eqref{equ:mRelationImP} holds.
    \item Design the function $\kappa$ solving~\eqref{equ:SimuDesignCond2}.
    \item Identify all differentiable left-inverse mappings $m$ of the mapping $p$.
    \item Find $m$ that satisfies condition~(ii). If not successful, return to step 1) to redesign the functions $p$, $l$, and $\bar{\phi}$.
    \item Design $\delta$ such that~\eqref{equ:mRelatSuffiDiff} and~\eqref{equ:mRelatSuffiIn} hold by following Theorem~\ref{thm:mRelateDelta}.
\end{enumerate}



\begin{remark}
    As required by Theorem~\ref{thm:mRelateDelta} and in step 3) above, the differentiability of the function $m$ is necessary for~\eqref{equ:mRelatSuffiDiff} to hold. Since we have shown, below Theorem~\ref{thm:mCondi}, that the function $m$ is the left inverse of the injective function $p$, the differentiability of $m$ can be guaranteed by selecting $p$ to be a smooth immersion/embedding, see~\cite[Chapter~4]{ref:lee2013introduction} for more details.
\end{remark}

The above steps only guarantee the satisfaction
of~\eqref{equ:SimuLocalRequire},~\eqref{equ:SimuDesign}, and~\eqref{equ:mRelatSuffi}. By Theorem~\ref{thm:SimuDesign} and Lemma~\ref{lem:mRelatePDE}, this means that
the $m$-relation is satisfied, but the bounded output discrepancy
is not guaranteed, as it requires extra conditions, \textit{i.e.}, Assumption~\ref{asmp:outputLipstchiz} and condition~\textbf{(B)} as in~\eqref{equ:SimuResBound}. By taking the simulation function into consideration, the final design procedure for solving Problem~\ref{prob:ASHC} is summarised in Fig.~\ref{fig:Flow_Chart}.



In fact, for a general nonlinear function $h$ of the concrete system $\Sigma$, deriving a constructive condition on the function $p$ in~\eqref{equ:SimuDesign} that guarantees the existence of a differentiable function $m$ satisfying Theorem~\ref{thm:mCondi} remains challenging. Nevertheless, it should be noted that in the context of ASHC, the $m$-relation requirement discussed in this subsection is not as critical as the bounded output discrepancy requirement in Section~\ref{subsec:BODdesign}. In fact, the bounded output discrepancy is the core target of the ASHC technique, while the $m$-relation is added by some authors as a requirement to provide a sufficient guarantee of the input-output controllability of the constructed abstract system. On the one hand, this controllability guarantee is a sufficient but not necessary condition for preserving the ability of system~$\Sigma^\prime$ to recover some designated output behaviour of system~$\Sigma$. Failing to satisfy the $m$-relation does not imply that system~$\Sigma$ is unable to satisfy the designated control goal. On the other hand, even if the control goal, in the worst case, cannot be satisfied \textit{exactly} due to the lack of this controllability, the abstract system~$\Sigma^\prime$ can still be controlled to \textit{approximately} recover the expected output trajectory with a guaranteed output error bound. Finally, note that as indicated by Theorems~\ref{thm:mRelateDelta} and~\ref{thm:mCondi}, although the condition~\eqref{equ:mRelatSuffiOut} is not straightforward to satisfy, the other conditions~\eqref{equ:mRelatSuffiDiff} and~\eqref{equ:mRelatSuffiIn} can be guaranteed by augmenting the rank of $\delta$, \textit{i.e.}, one can improve the input-output controllability by enlarging the input space of the abstract system~$\Sigma^\prime$ (to even cover the entire state space). 

For the above reasons, most existing literature in the ASHC field omits the study of $m$-relation.

\section{Illustrative Example}\label{sec:example}
In this section, we illustrate the results presented in this paper with an example. To this end, we consider the ASHC problem of a DC-to-DC Ćuk converter, and we solve the problem following the design procedures shown in Fig.~\ref{fig:Flow_Chart}.

\begin{figure}[tbp]
\begin{centering}
    \includegraphics[width=0.92\linewidth]{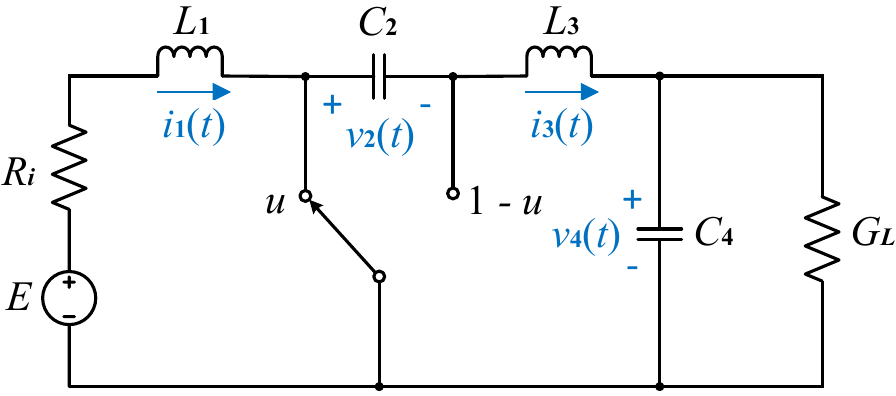}
    \caption{DC-to-DC Ćuk Converter.}
    \label{fig:Converter_Circuit}
\end{centering}
\end{figure}
\subsection{System Modelling}
Consider a DC-to-DC Ćuk converter depicted in Fig.~\ref{fig:Converter_Circuit}, where $E$ represents the value of constant input voltage, $R_i$ the input resistance, and $G_L$ the load admittance. The converter consists of two inductors, $L_1$ and $L_3$, and two capacitors, $C_2$ and $C_4$. The currents flowing through the two inductors are indicated with $i_1$ and $i_3$, respectively, while the voltages across the two capacitors are denoted by $v_2$ and $v_4$, respectively. The input $u$ denotes the duty cycle of a PWM signal controlling the position of the switch in the converter, and therefore satisfies $0 \leq u \leq 1$. The averaged model of this converter circuit is given by the following equations.
\begin{equation}\label{equ:ConverterAveMod}
    \begin{aligned}
    L_1 \frac{d}{d t} i_1 & = -(1-u) v_2 + E - i_1 R_i, \\
    C_2 \frac{d}{d t} v_2 & = (1-u) i_1 + u i_3, \\
    L_3 \frac{d}{d t} i_3 & = -u v_2 - v_4, \\
    C_4 \frac{d}{d t} v_4 & = i_3 - G_L v_4.
    \end{aligned}
\end{equation}
Now denote $\mathbf{x} = [\mathbf{x}_1, \mathbf{x}_2, \mathbf{x}_3, \mathbf{x}_4]^\top$ with $\mathbf{x}_1 = i_1$, $\mathbf{x}_2 = v_2$, $\mathbf{x}_3 = i_3$, and $\mathbf{x}_4 = v_4$ as the states, and let $\mathbf{y} = v_4$ as the output. Then the averaged model~\eqref{equ:ConverterAveMod} can be modelled as a (concrete) nonlinear system~\eqref{equ:ConcreteSystem} satisfying the input-affine property~\eqref{equ:concreteAffine}, where
\begin{equation}\label{equ:ConvertSystem}
\begin{aligned}
    \bar{f}(x) \!&=\! \left[\!\begin{array}{c}
        \frac{1}{L_1}(-R_i x_1 \!-\! x_2 \!+\! E)  \\
        \frac{1}{C_2} x_1 \\
        -\frac{1}{L_3} x_4 \\
        \frac{1}{C_4}(x_3 \!-\! G_L x_4)
    \end{array}\!\right]\!\!, \quad 
    g(x) \!= \!\left[\!\!\begin{array}{c}
        \frac{1}{L_1} x_2  \\
        \frac{1}{C_2} (-\!x_1 \!+\! x_3) \\
        -\frac{1}{L_3} x_4 \\
        0
    \end{array}\!\!\right]\!\!, \\
    h(x) \!&= x_4 = C x,   \qquad C = \left[\begin{array}{cccc}
        0  &  0  &  0  &  1
    \end{array}\right],
\end{aligned}
\end{equation}
with dimensions $n = 4$, $m = p = 1$, and input space $\mathbf{U} = [0, 1]$. Note that the output mapping $h$ satisfies Assumption~\ref{asmp:outputLipstchiz}. Furthermore, this concrete system can alternatively be formulated as
\begin{equation}\label{equ:ConvertSystemLinear}
    f(x, u) = \bar{A}(u)x + \bar{b}
\end{equation}
with 
\begin{equation*}
   \bar{A}(u) \!=\! \left[\!\begin{array}{cccc}
        -\frac{R_i}{L_1}  &  -\frac{1-u}{L_1}  &  0  &  0 \\
        \frac{1-u}{C_2}  &  0  &  \frac{u}{C_2}  &  0 \\
        0  &  -\frac{u}{L_3}  &  0  &  -\frac{1}{L_3} \\
        0  &  0  &  \frac{1}{C_4}  &  -\frac{G_l}{C_4}
    \end{array}\!\right]\!\!, \quad 
    \bar{b} \!= \!\left[\!\begin{array}{c}
        \frac{E}{L_1}  \\
        0 \\
        0 \\
        0
    \end{array}\!\right]\!\!,
\end{equation*}
where $\bar{A}$ is a linear function of $u$. This formulation~\eqref{equ:ConvertSystemLinear} will be used later. In this example, we arbitrarily set the values of the electrical components as listed in Table~\ref{tab:parameter}. With such a selection of parameters, the converter is expected to output voltages between 0 and -120, \textit{i.e.}, $\mathbf{Y} = [-120, 0]$. Consequently, as $y = h(x) = x_4$, the state is expected to stay within the set $\mathbf{X}_y = \{ x \in \mathbb{R}^n \,|\, h(x) \in \mathbf{Y}\} = \mathbb{R}^3 \times [-120, 0]$.

\begin{table}[tbp]
    \centering
    \begin{tabular}{|c|c|c|c|}
    \hline $R_i = 0.05\,\Omega$ & $L_1=10\, \mathrm{mH}$ & $C_2= 11\,\mathrm{mF}$ &  $E=12\,\mathrm{V}$ \\
    \hline $G_L=44.7\,\mathrm{mS}$ & $L_3=10\, \mathrm{mH}$ & $C_4= 11 \,\mathrm{mF}$ & \\
    \hline
    \end{tabular}
    \caption{Parameters of the DC-to-DC Ćuk converter model.}
    \label{tab:parameter}
\end{table}

We now show that this converter system is asymptotically stabilisable. Since the concrete system is dissipative and can be formulated as~\eqref{equ:ConvertSystemLinear}, following~\cite[Theorem 3]{ref:girard2007approximate}, there exists a matrix $M \in \mathbb{R}^{n \times n}$ and a strictly positive scalar $\lambda$ such that
\begin{equation}\label{equ:egLMI}
    \bar{A}^\top(u) M + M \bar{A}(u) \preceq -\lambda M, \qquad  M \succ \bm{0}_{n \times n}.
\end{equation}
for all $u \in [0, 1]$. Since $\bar{A}$ is a linear function of $u$, the matrix $M$ can be computed, with parameter values in Table~\ref{tab:parameter}, by solving the linear matrix inequalities~\eqref{equ:egLMI} for both $u = 0$ and $u = 1$. Using the solver MOSEK~\cite{ref:mosek} with the YALMIP toolbox~\cite{ref:YALMIP}, we obtain
\begin{equation}\label{equ:egM}
M = \left[\begin{array}{rrrr}
0.4804 & 0.0102 & 0.0002 & -0.0093 \\
0.0102 & 0.5304 & 0.0081 & 0.0001 \\
0.0002 & 0.0081 & 0.4824 & -0.0135 \\
-0.0093 & 0.0001 & -0.0135 & 0.5304
\end{array}\right]
\end{equation}
with $\lambda = 2$. As a consequence, now we can show that the stabilisability Lyapunov function
\begin{equation}\label{equ:egV}
    V(x, z) = \|x - z\|_M^2 = (x - z)^\top M (x - z)
\end{equation}
satisfies~\eqref{equ:stabilisableV} with $k(z, x, u) = u$. Let $\underline{\sigma}$ and $\overline{\sigma}$ be the smallest and largest eigenvalue of $M$ in~\eqref{equ:egM}, respectively. Then $V$ in~\eqref{equ:egV} satisfies~\eqref{equ:stabilisableCond1} as $\underline{\sigma} \|x - z\|^2 \leq V \leq \overline{\sigma} \|x - z\|^2$. Then to show~\eqref{equ:stabilisableCond2}, we compute
\begin{equation*}
\begin{aligned}
    \dot{V}(x, z) &\!=\! (x \!-\! z)\!^\top M(\dot{x} \!-\! \dot{z}) \!=\! 2(x \!-\! z)\!^\top M\left(\bar{A}(u)x \!-\! \bar{A}(u)z\right) \\
    &\!= 2(x - z)\!^\top M\bar{A}(u)(x - z).
\end{aligned}
\end{equation*}
Since $M$ satisfies~\eqref{equ:egLMI}, the last equation yields that $\dot{V}(x, z) \leq -\lambda V \leq -\lambda \underline{\sigma} \|x - z\|^2 $, \textit{i.e.},~\eqref{equ:stabilisableCond2} holds. In addition, this Lyapunov function $V$ is quadratic and satisfies~\eqref{equ:bounddVdx}, which, by Proposition~\ref{prop:condB}, can guarantee the satisfaction of condition~\textbf{(B)} if $\delta$ in~\eqref{equ:abstractAffine} is designed to be continuous.

\subsection{Solving~(\ref{equ:SimuDesign}) and~(\ref{equ:mRelatSuffi})}
Following the procedures shown in Fig.~\ref{fig:Flow_Chart}, the next step is to select $p$, $l$, and $\bar{\phi}$ solving~\eqref{equ:SimuDesignCond1}. Here we select
\begin{equation}\label{equ:egSimuDesign}
    \bar{\phi} = 0, \quad p(\xi) = 
    \left[\!\begin{array}{c}
    \frac{EG_L \xi^2 }{{\left(\xi -1\right)}^2 +G_L R_i \xi^2 }\\
    -\frac{E{\left(\xi -1\right)}}{{\left(\xi -1\right)}^2 +G_L R_i \xi^2 }\\
    \frac{EG_L \xi {\left(\xi -1\right)}}{{\left(\xi -1\right)}^2 +G_L R_i \xi^2 }\\
    \frac{E\xi {\left(\xi -1\right)}}{{\left(\xi -1\right)}^2 +G_L R_i \xi^2 }
    \end{array}\!\right], \quad l(\xi) = \xi,
\end{equation}
for all $\xi \in \mathbb{R}^{\hat{n}}$ with $\hat{n} = 1$. Note that with $\mathbf{X}_y = \mathbb{R}^3 \times [-120, 0]$ and $h$ in~\eqref{equ:ConvertSystem}, the function $p$ in~\eqref{equ:egSimuDesign} satisfies~\eqref{equ:mRelationImP}, but is not injective. We show later that $p$ is injective over a certain region of interest $\mathbf{V} \subset \mathbb{R}^{\hat{n}}$, and is such that condition~(ii) in Theorem~\ref{thm:mCondi} holds.
With the selection in~\eqref{equ:egSimuDesign}, we obtain
\begin{equation}\label{equ:egKappa}
    \kappa(\xi) = h(p(\xi)) = \frac{E\,\xi {\left(\xi -1\right)}}{{{\left(\xi -1\right)}}^2 +G_L R_i \xi^2 },
\end{equation}
which solves~\eqref{equ:SimuDesignCond2}. 
Note that in this case, the interface function in~\eqref{equ:Interface} satisfies $u_w = \xi + u^*(\xi, x, v)$ as $k(z, x, u) = u$. Since the input space $\mathbf{U} = [0, 1]$, we require $u_w = \xi + u^*(\xi, x, v) \in [0, 1]$.

After solving~\eqref{equ:SimuDesign}, we focus on the satisfaction of the $m$-relation formulated in~\eqref{equ:mRelatSuffi}. According to Fig.~\ref{fig:Flow_Chart}, we first find $m$ satisfying condition~(i) in Theorem~\ref{thm:mCondi}, which is an intermediate step for solving~\eqref{equ:mRelatSuffiOut}. To this end, consider that~\eqref{equ:mRelatSuffiOut} requires $h(x) = \kappa(m(x)) = x_4$ with $\kappa$ in~\eqref{equ:egKappa}. By simplifying this equation, we obtain 
\begin{equation}\label{equ:mEquation}
    [(G_L R_i + 1) x_4 - E] m^2 + (E - 2x_4) m + x_4 = 0,
\end{equation}
which is solvable whenever $x_4 \in [-\frac{E}{2\sqrt{R_i G_L}}, \frac{E}{2\sqrt{R_i G_L}}]$. By substituting the values in Table~\ref{tab:parameter}, we get $[-\frac{E}{2\sqrt{R_i G_L}}, \frac{E}{2\sqrt{R_i G_L}}] = [-126.91, 126.91] \supset \mathbf{Y} = [-120, 0]$. Recall that $\mathbf{X}_y = \mathbb{R}^3 \times [-120, 0]$. This means that when the converter is operating on the state set $\mathbf{X}_y$ and generating the expected output $\mathbf{y}(t) \in \mathbf{Y}$,~\eqref{equ:mEquation} is always solvable. Note that~\eqref{equ:mEquation} has two analytical solutions as the function $p$ (and consequently $\kappa(\xi) = h(p(\xi))$) is not injective for all $\xi \in [0, 1]$, as mentioned earlier and shown in Fig.~\ref{fig:kappa_xi}. However, this figure indicates that to generate outputs $y$ in $\mathbf{Y} = [-120, 0]$, it suffices to only focus on $\xi \in [0, 0.95]$. Thus, we restrict the domain to $\mathbf{V} = [0, 0.95]$, which satisfies~\eqref{equ:SimuLocalRequire} and is such that the function $p$, as well as $\kappa = h \circ p$, is injective. Then, we determine
\begin{equation}\label{equ:egm}
    m(x) = \frac{-E + 2x_4 + \sqrt{{E}^2 -4G_L R_i {x_4 }^2 }}{2(G_L R_i + 1) x_4 - 2E },
\end{equation}
which is the left inverse function of $p$ on $\mathbf{V} = [0, 0.95]$ as $m(p(\xi)) = \xi$. In fact, this function $m$ also satisfies condition~(ii) as
\begin{equation*}
    x - p(m(x)) = [*_1, *_2, *_3, 0] \in \operatorname{Ker}(h),
\end{equation*}
for some values $*_1, *_2, *_3 \in \mathbb{R}$. Then the satisfaction of condition~(ii) follows from Corollary~\ref{corol:mRelateGeo}. 
Finally, by Theorem~\ref{thm:mRelateDelta}, we can simply have $\delta = 1$, which guarantees the satisfaction of~\eqref{equ:mRelatSuffiDiff} and~\eqref{equ:mRelatSuffiIn}. Note also that by Proposition~\ref{prop:condB}, this design of $\delta$ ensures the satisfaction of condition~\textbf{(B)} in Theorem~\ref{thm:SimuDesign}, thereby guaranteeing the existence of a simulation function. This concludes the design. The obtained abstract system solves Problem~\ref{prob:ASHC}.

\begin{figure}[tbp]
\begin{centering}
    \includegraphics[width=\linewidth]{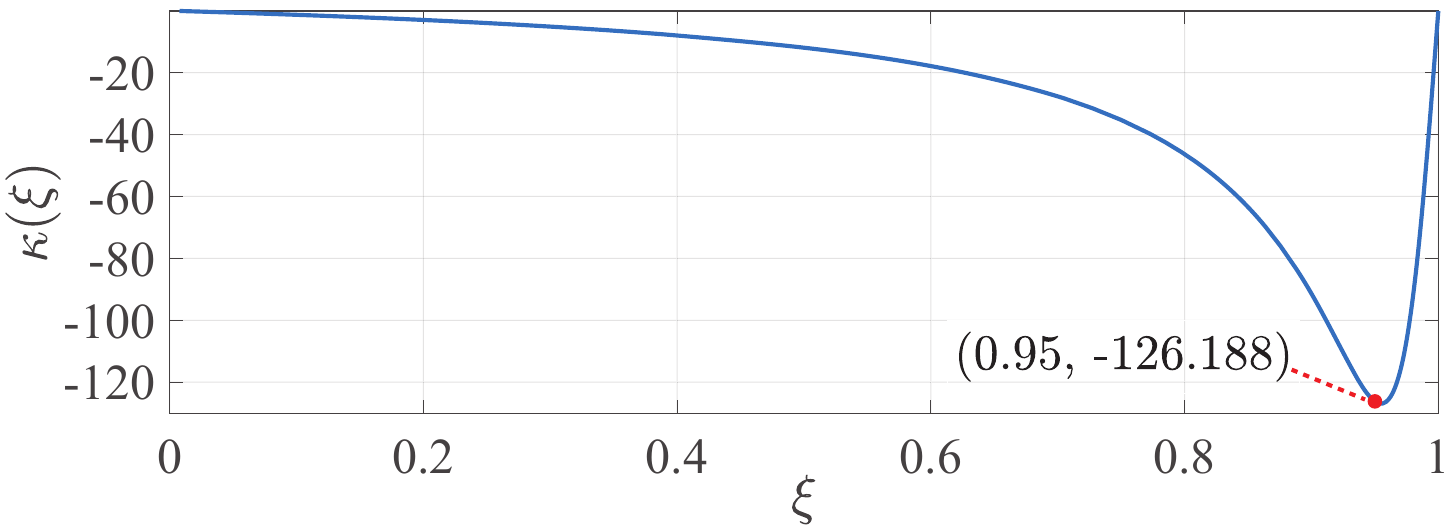}
    \caption{Plot of $\kappa(\xi)$ versus $\xi \in [0, 1]$ with parameter values in Table~\ref{tab:parameter}.}
    \label{fig:kappa_xi}
\end{centering}
\end{figure}

\subsection{Simulation Function and Error Bound}
By following Fig.~\ref{fig:Flow_Chart}, we have constructed an ASHC framework, and the final step is to characterise and reduce the error bound in~\eqref{equ:ErrorBound} by studying the simulation function. With $V$ in~\eqref{equ:egV}, the candidate simulation function takes the form 
\begin{equation}\label{equ:egSimuFunc}
    W(\xi, x) = V(p(\xi), x) = (p(\xi) - x)\!^\top M (p(\xi) - x).
\end{equation}
By Theorem~\ref{thm:SimuDesign}, this simulation function satisfies~\eqref{equ:SimuFuncCondi1}. To check~\eqref{equ:SimuFuncCondi2}, we first note that by straightforward computations, we have
\begin{equation}\label{equ:egConditionB}
\begin{aligned}
    &\!\!\!\frac{\partial W}{\partial p(\xi)} \frac{\partial p(\xi)}{\partial \xi} \delta(\xi)v + \frac{\partial W}{\partial x} g(x) u^* \\
    &\!\!\!= 2(p(\xi) - x)^\top M \left(\frac{\partial p(\xi)}{\partial \xi} \delta(\xi)v - g(x) u^*\right) \\
    &\!\!\!\leq \varepsilon \|p(\xi) - x\|_M^2 + \frac{1}{\varepsilon}\left\|\sqrt{M} \left(\frac{\partial p(\xi)}{\partial \xi} \delta(\xi)v - g(x) u^*\!\right)\right\|^2 \\
    &\!\!\!\leq \varepsilon \|p(\xi) \!-\! x\|_M^2 \!+\! \frac{1}{\varepsilon}\left\|\sqrt{M} \left(\frac{\partial p(\xi)}{\partial \xi} \delta(\xi) \!-\! g(x) q(\xi, x)\!\right)\right\|^2 \!\!\!\|v\|^2,
\end{aligned}
\end{equation}
where we have used the same simplification steps as in~\eqref{equ:SimuResBoundYoung} with $\varepsilon > 0$, and we have designed $u^*$ as in~\eqref{equ:uMinBound}. Then we obtain
\begin{equation*}
\begin{aligned}
    &\nabla W(\xi, x)^T\left[\begin{array}{c}\phi(\xi, v) \\
    f(x, u_w(\xi, x, v))
    \end{array}\right] \\
    &=2\|p(\xi) - x\|_{M \bar{A}(l(\xi))}^2 + \frac{\partial W}{\partial p(\xi)} \frac{\partial p(\xi)}{\partial \xi} \delta(\xi)v + \frac{\partial W}{\partial x} g(x) u^* \\
    &\leq \!-(\lambda \!-\! \varepsilon) W\!(\xi, x) \!+\! \frac{1}{\varepsilon}\!\left\|\sqrt{M}\! \left(\!\!\frac{\partial p(\xi)}{\partial \xi} \delta(\xi) \!-\! g(x) q(\xi, x)\!\!\right)\!\right\|^2 \!\!\!\!\|v\|^2\!.
\end{aligned}
\end{equation*}
In this manner, with any $\varepsilon < \lambda$,~\eqref{equ:SimuFuncCondi2} holds if the function $q$ is designed such that there exists a positive constant $\bar{d}$ satisfying 
$\|\vartheta(\xi, x)\| \leq \bar{d}$ with
\begin{equation}\label{equ:iota}
    \vartheta(\xi, x) := \sqrt{M} \left(\!\frac{\partial p(\xi)}{\partial \xi} \delta(\xi) \!-\! g(x) q(\xi, x)\!\right)
\end{equation}
for all $(\xi, x) \in \mathbf{V} \times \mathbb{R}^n$. In fact, as shown by Proposition~\ref{prop:condB}, this constant $\bar{d}$ exists even when $q \equiv 0$, as the term $\|\frac{\partial p(\xi)}{\partial \xi} \delta(\xi)\|$ with $p$ in~\eqref{equ:egSimuDesign} and $\delta \equiv 1$ is bounded for $\xi \in \mathbf{V} = [0, 0.95]$. With parameter values in Table~\ref{tab:parameter}, this boundedness is illustrated in Fig.~\ref{fig:delta_Compare} (left), which displays the norm of $\vartheta(\xi, x)$ with $q \equiv 0$, showing that the constant $\bar{d} = 1760$ satisfies $\|\vartheta(\xi, x)\| \leq \bar{d}$ for all $(\xi, x) \in \mathbf{V} \times \mathbb{R}^n$. 

\begin{figure}[tbp]
\begin{centering}
    \includegraphics[width=\linewidth]{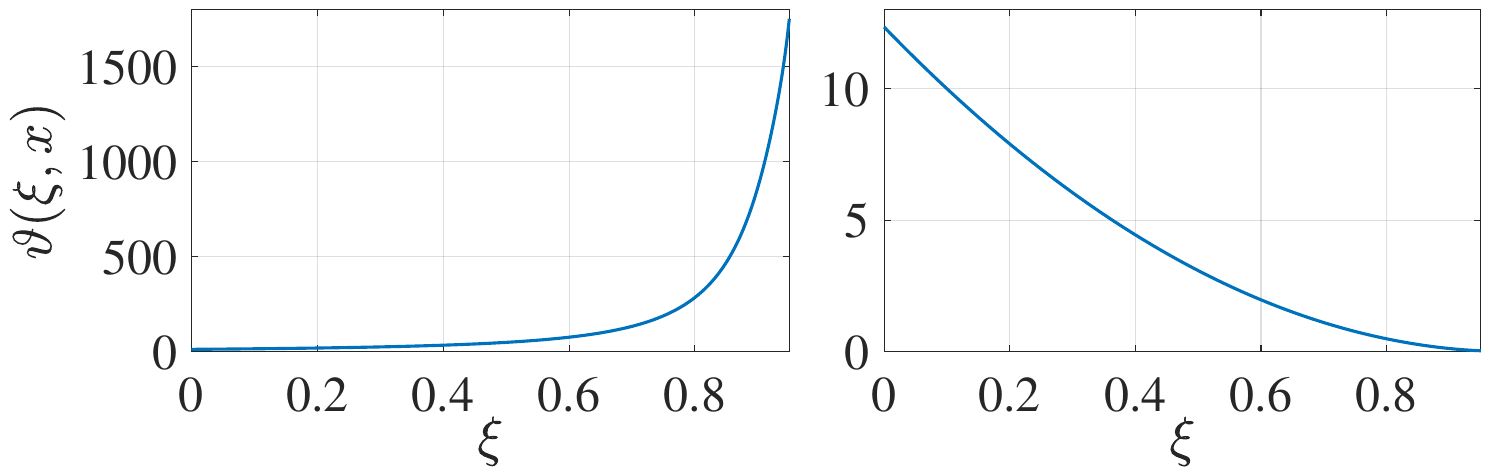}
    \caption{Plots of $\|\vartheta(\xi, x)\|$ versus $\xi \in [0,0.95]$ when $q \equiv 0$, with $\delta = 1$ (left) and $\delta$ in~\eqref{equ:egDelta} (right).}
    \label{fig:delta_Compare}
\end{centering}
\end{figure}

While the above designs solve the ASHC problem, the bound of the output discrepancy in~\eqref{equ:ErrorBound} can be reduced. To this end, note that,~\eqref{equ:egConditionB} implies that $\eta(r) = (\lambda - \varepsilon)r$ and $\gamma(r) = \frac{1}{\varepsilon}\|\vartheta(\xi, x)\|^2\|r\|^2$ in~\eqref{equ:ErrorBound} with $\alpha_h(r) = 0.52r^2$ satisfying~\eqref{equ:SimuFuncCondi1}. This can be seen by the fact that since $M$ in~\eqref{equ:egM} satisfies $M \succeq 0.52C^\top C$ with $C$ in~\eqref{equ:ConvertSystem}, $W(\xi, x) = (p(\xi) - x)\!^\top M (p(\xi) - x) \geq 0.52 (p(\xi) - x)\!^\top C^\top C (p(\xi) - x) = 0.52\|\kappa(\xi) - h(x)\|^2$. Consequently, 
\begin{equation}\label{equ:egErrorBound}
    \alpha_h^{-1}(\eta^{-1}(2\gamma(\|\mathbf{v}\|_{\infty}))) = \sqrt{\frac{2}{0.52(\lambda - \varepsilon)\varepsilon}}\|\vartheta(\xi, x)\| \|\mathbf{v}\|_\infty. 
\end{equation}
Therefore, reducing the bound~\eqref{equ:ErrorBound} boils down to minimising the term $\sqrt{\frac{1}{0.26(\lambda - \varepsilon)\varepsilon}}\|\vartheta(\xi, x)\|$, in which the first multiplier is minimised when $\varepsilon = 1$ as $\lambda = 2$. Then the task comes down to minimising the term $\|\vartheta(\xi, x)\|$.

One way of minimising $\|\vartheta(\xi, x)\|$ is by redesigning $\delta$. Previously, with $\delta = 1$, Fig.~\ref{fig:delta_Compare} shows that $\|\vartheta(\xi, x)\|$ increases fast when $\xi$ is close to 0.95. By observing the function $\frac{\partial p(\xi)}{\partial \xi}$ (obtainable from~\eqref{equ:egSimuDesign}), this growth can be reduced by redesigning $\delta$ as
\begin{equation}\label{equ:egDelta}
    \delta(\xi) = ({{\left(\xi -1\right)}}^2 +G_L R_i \,\xi^2 )^2.
\end{equation}
Consequently, as shown by Fig.~\ref{fig:delta_Compare} (right), $\|\vartheta(\xi, x)\|$ with $q \equiv 0$ is greatly reduced, with the bound $\bar{d} = 12.35$. Such a modification in $\delta$ brings two advantages. On the one hand, the reduction in $\|\vartheta(\xi, x)\|$ allows more flexibility in designing the input $\mathbf{v}$ that controls the abstract system, as $\|\mathbf{v}\|_{\infty}$ can be larger while resulting in the same bound~\eqref{equ:egErrorBound} as before. On the other hand, as shown in Fig.~\ref{fig:delta_Compare}, $\delta$ in~\eqref{equ:egDelta} leads to a flatter curve of $\vartheta(\xi, x)$. Then the derived bound~\eqref{equ:egErrorBound}, computed by taking the maximum value $\bar{d}$ of $\vartheta(\xi, x)$, yields a less conservative estimate. Note that this new design of $\delta$ does not violate the satisfaction of the $m$-relation, as $\delta(\xi)$ in~\eqref{equ:egDelta} is invertible for all $\xi \in \mathbf{V}$. In fact, by the proof of Theorem~\ref{thm:mRelateDelta}, the functions $b$ and $c$ solving~\eqref{equ:mRelatSuffiDiff}--\eqref{equ:mRelatSuffiIn} are of the form
\begin{equation}\label{equ:egMrelateLink}
\begin{aligned}
    b(x) &= -\frac{2{\left(x_3 -G_L x_4 \right)}{{\left(x_4 -E+G_L R_i x_4 \right)}}^2 }{C_4 E \chi {\left(E+\chi-G_L R_i \chi+E G_L R_i -4G_L R_i x_4 \right)}}, \\
    c(x) &= 0,
\end{aligned}
\end{equation}
where $\chi = {E}^2 -4G_L R_i {x_4}^2$.

Another way of minimising $\|\vartheta(\xi, x)\|$, as mentioned in Section~\ref{subsec:BODdesign}, is to manipulate $u^*$. In this example, with $u^*$ in~\eqref{equ:uMinBound} and $\vartheta$ in~\eqref{equ:iota}, the minimisation of  $\|\vartheta(\xi, x)\|$ becomes a least square approximation problem. Whenever $\operatorname{rank} (g(x)) = 1$, \textit{i.e.}, $g(x) \neq \bm{0}_{4 \times1}$, the minimum of $\|\vartheta(\xi, x)\|$ is reached by having $q(\xi, x) = (g\!^\top\!(x) M g(x))^{-1}g\!^\top\!(x) M\frac{\partial p(\xi)}{\partial \xi} \delta(\xi)$. This result is a direct extension of~\cite[Proposition 1]{ref:girard2009hierarchical} to the nonlinear case.

In summary, to solve the ASHC problem of the converter, we can design the abstract system~$\Sigma^\prime$ as
\begin{equation}\label{equ:egAbstract}
\Sigma^\prime:
\left\{\begin{aligned}
    \dot{\bm{\xi}} &= ({{\left(\bm{\xi} -1\right)}}^2 +G_L R_i \bm{\xi}^2 )^2 \mathbf{v}, \\
    \bm{\psi} & = \frac{E\bm{\xi} {\left(\bm{\xi} -1\right)}}{{{\left(\bm{\xi} -1\right)}}^2 +G_L R_i \bm{\xi}^2 },
\end{aligned}\right.
\end{equation}
and the interface function $u_w$ takes the form
\begin{equation}\label{equ:egInterface}
    u_w(\xi, x, v) = \xi + q(\xi, x) v,
\end{equation}
with
\begin{equation}\label{equ:egqLS}
    q(\xi, x) \!=\! \left\{\!\!\! 
    \begin{array}{cc}
        (g\!^\top\!(x) M g(x))^{-\!1}g\!^\top\!(x)M\frac{\partial p(\xi)}{\partial \xi} \delta(\xi), \!& g(x) \neq \bm{0}_{4 \times 1}; \\
        0, \!& g(x) = \bm{0}_{4 \times 1}.
    \end{array}\right.
\end{equation}

\begin{remark}
   Recall that the input space is $\mathbf{U} = [0, 1]$. There may not be a guarantee on the fact that $u_\omega$ in~\eqref{equ:egInterface} stays inside $\mathbf{U}$ for all times.
   In such a case, one can instead design $u_\omega$ as \begin{equation}\label{equ:egInterfaceSat}
       u_\omega = \operatorname{sat}\!\left(\xi + q(\xi, x) v\right),
    \end{equation}
   where the saturation function $\operatorname{sat}(r) := \max(0, \min(1, r))$. Since $\xi$ always stays within $\mathbf{V} = [0, 0.95] \subset \mathbf{U}$, the design in~\eqref{equ:egInterfaceSat} can be regarded as adding a limitation on the norm of $q(\xi, x)$ for minimising the output discrepancy, which does not violate the bounded discrepancy result established before.
\end{remark}


\begin{figure}[tbp]
\begin{centering}
    \includegraphics[width=\linewidth]{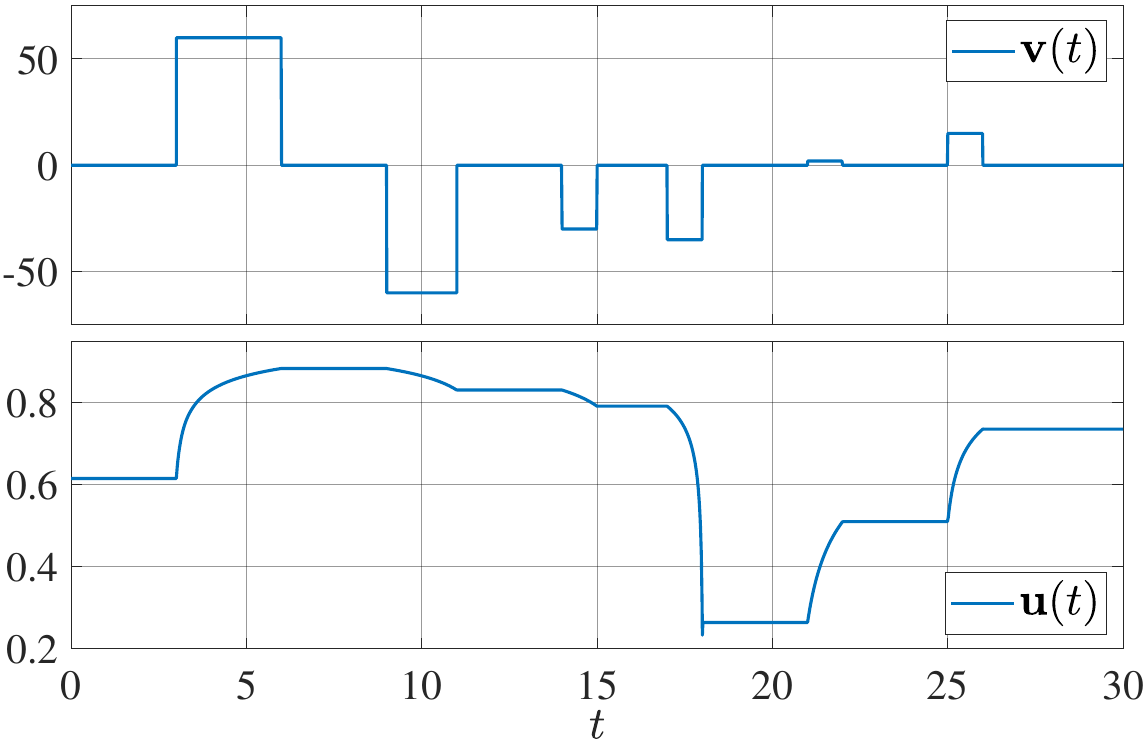}
    \caption{Time histories of inputs $\mathbf{v}$ (top) and $\mathbf{u}$ (bottom) driving systems~$\Sigma^\prime$ and $\Sigma$, respectively, in the ASHC framework.}
    \label{fig:Example_BoundedIn}
\end{centering}
\end{figure}

\subsection{Simulation Result}
Now we are ready to show the simulation results of the ASHC technique based on the proposed abstract system~\eqref{equ:egAbstract} and interface~\eqref{equ:egInterface}. We first illustrate the bounded output discrepancy property. Then we show the satisfaction of the $m$-relation. All simulations are conducted using MATLAB.

In this example, the objective is to control the concrete system~$\Sigma$ (in this case, the converter~\eqref{equ:ConverterAveMod}) to output the specific voltage levels as $$[-19.11, -80.90, -44.27, -4.31, -12.48, -32.91].$$ To this end, we control the abstract system~$\Sigma^\prime$ and synthesise the controller of the concrete system~$\Sigma$ via the ASHC technique. The initial condition of the abstract system~$\Sigma^\prime$ is randomly selected as $\bm{\xi}(t_0) = \xi_0 = 0.6156$, while that of the concrete system~$\Sigma$ is set as $\textbf{x}(t_0) = x_0 = p(\xi_0) = [1.3678, 31.0396, -0.8541, -19.1080]^\top$, where $t_0  = 0$ in this simulation. Note that this special selection of $\textbf{x}(t_0)$ is not necessary, and is performed here only for illustrative purposes to be explained later. The simulation results of the bounded output discrepancy are displayed by Figs.~\ref{fig:Example_BoundedIn} and~\ref{fig:Example_BoundedOut}. Fig.~\ref{fig:Example_BoundedIn} shows the time histories of the input $\mathbf{v}$ driving the abstract system~$\Sigma^\prime$ (top) and the input $\mathbf{u}$ that controls the concrete system~$\Sigma$ (bottom).  Fig.~\ref{fig:Example_BoundedOut} depicts the time histories of the outputs $\mathbf{y}$ and $\bm{\psi}$ of systems~$\Sigma$ and~$\Sigma^\prime$, respectively, (top), and the time history of their error $\mathbf{e_y} = \bm{\psi} - \mathbf{y}$ (bottom). The results confirm that the control input $\mathbf{u}$ varies within $\mathbf{U} = [0, 1]$, and the output error between the two systems is bounded, implying that the concrete system $\Sigma$ approximately achieves its control objective. 
Note also that by having the initial condition $x_0 = p(\xi_0)$, the ASHC problem starts with the simulation function $W(\xi, x) = 0$, meaning that the first term of the error bound in~\eqref{equ:ErrorBound} is eliminated. Hence, the error bound only depends on $\|\mathbf{v}\|_\infty$. 

\begin{figure}[tbp]
\begin{centering}
    \includegraphics[width=\linewidth]{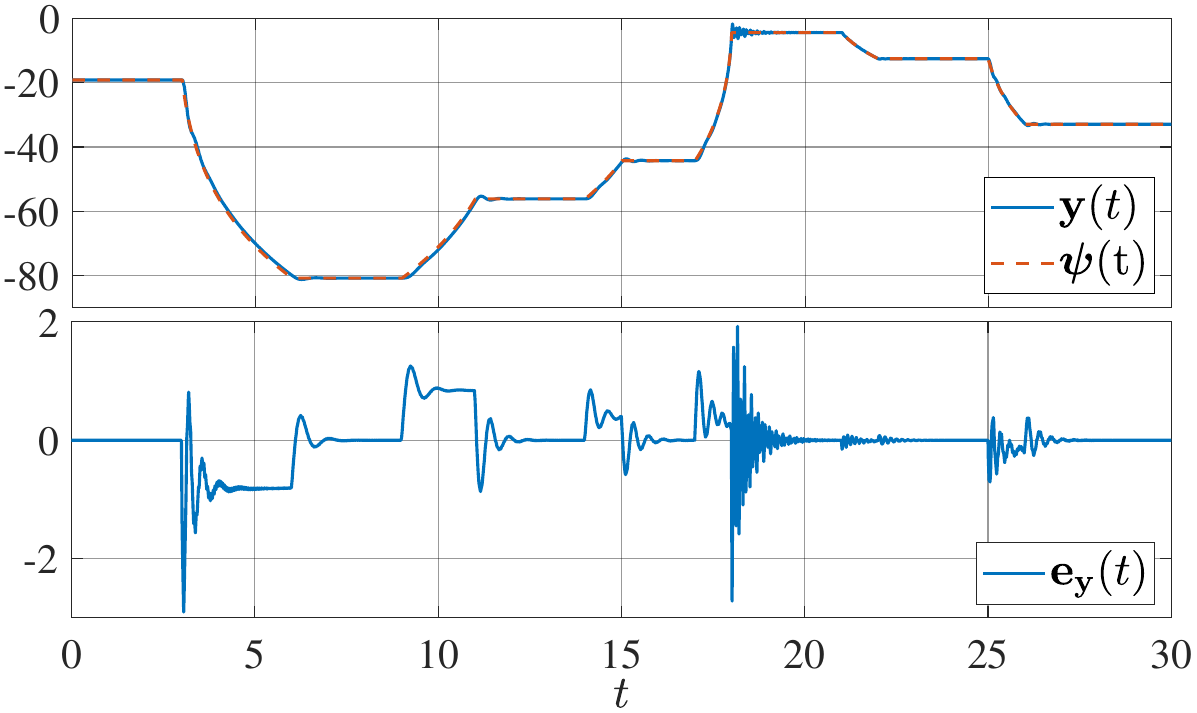}
    \caption{Time histories of outputs $\mathbf{y}$ (blue solid) and $\bm{\psi}$ (red dashed) of systems~$\Sigma$ and~$\Sigma^\prime$ (top) and their error $\mathbf{e_y} = \bm{\psi} - \mathbf{y}$ (bottom).}
    \label{fig:Example_BoundedOut}
\end{centering}
\end{figure}

\begin{remark}
    In this example, as $\|\mathbf{v}\|_\infty = 60$, the computation of the error bound via~\eqref{equ:egErrorBound} yields $\|\mathbf{e_y}\| \leq \sqrt{\frac{1}{0.26}}\times 12.35 \times 60 = 1453.22$, which is not tight. A reason for this is that the bound~\eqref{equ:egErrorBound} is computed at the maximum point of $\|\vartheta(\xi, x)\|$ corresponding to $\xi=0.95$, see Fig.~\ref{fig:delta_Compare} (right). However, the state of system~$\Sigma^\prime$ generally does not stay around this extreme point, and the real error is much smaller. Another reason is the use of $u^*$ for minimising the bound. The bound is computed according to $\|\vartheta(\xi, x)\|$ in Fig.~\ref{fig:delta_Compare} (right) with $q \equiv 0$ and therefore $u^* \equiv 0$. However, we have designed $q$ as in~\eqref{equ:egqLS}, so $\|\vartheta(\xi, x)\|$ is further reduced. This reduction is not taken into account when computing the error bound.
\end{remark}

\begin{figure}[tbp]
\begin{centering}
    \includegraphics[width=\linewidth]{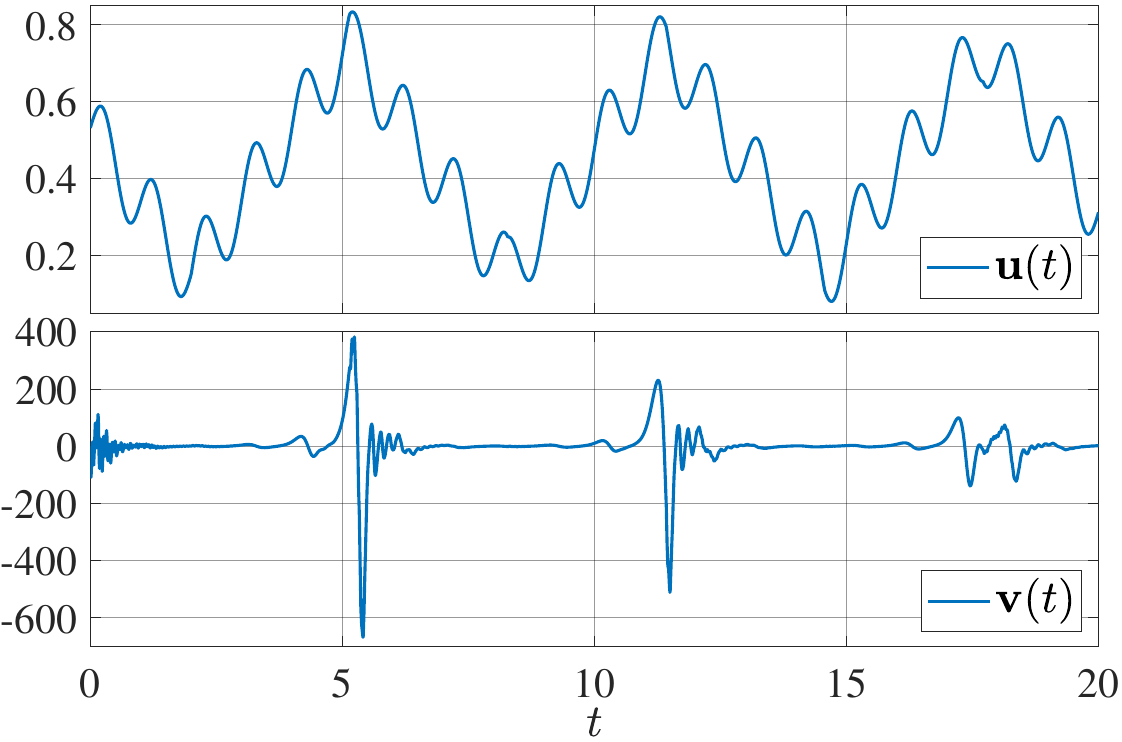}
    \caption{Time histories of inputs $\mathbf{u}$ (top) and $\mathbf{v}$ (bottom) into systems~$\Sigma$ and $\Sigma^\prime$, respectively, in the interconnection in Fig.~\ref{fig:mRelateInterconnect}.}
    \label{fig:Example_mRelationIn}
\end{centering}
\end{figure}

We then illustrate the satisfaction of the $m$-relation by implementing the interconnection in Fig.~\ref{fig:mRelateInterconnect}, where the functions $b$ and $c$ for constructing the link~\eqref{equ:mRelateLink} are as~\eqref{equ:egMrelateLink}. In this case, following the proof of Theorem~\ref{thm:mRelatOutMatch}, systems~$\Sigma$ and~$\Sigma^\prime$ are initialised on the manifold $\mathcal{M}_s = \{(x, \xi) \in (\mathbf{X}_y \times \mathbb{R}^{\hat{n}}) \,|\, \xi = m(x) \}$, with a randomly selected $$\mathbf{x}(t_0) = x_0 = [10.3256, 2.0561, -4.9785, -6.9732]^\top$$ and with $\bm{\xi}(t_0) = \xi_0 = m(x_0) = 0.3677$. The input $u$ is arbitrarily selected as $u = 0.3\nabla(t-2) + 0.1\sin(2\pi(t - 2)) + 0.45$, where $\nabla(t) \triangleq \frac{2}{\pi}\int_0^t \operatorname{sign}(\sin (\tau))d\tau-1$. Fig.~\ref{fig:Example_mRelationIn} displays the time histories of the input $u$ of system~$\Sigma$ (top) and the input $v$ in~\eqref{equ:mRelateLink} of system~$\Sigma^\prime$ (bottom). As a result, Fig.~\ref{fig:Example_mRelationOut} displays the time histories of the outputs $\mathbf{y}$ (solid blue) and $\bm{\psi}$ (red dashed) of systems~$\Sigma$ and $\Sigma^\prime$, respectively, showing that the outputs of two systems exactly match each other for all times. In other words, this result suggests that the $m$-relation successfully preserves the ability of the abstract system to recover the input-output controllability of the concrete system. 

\begin{figure}[tbp]
\begin{centering}
    \includegraphics[width=\linewidth]{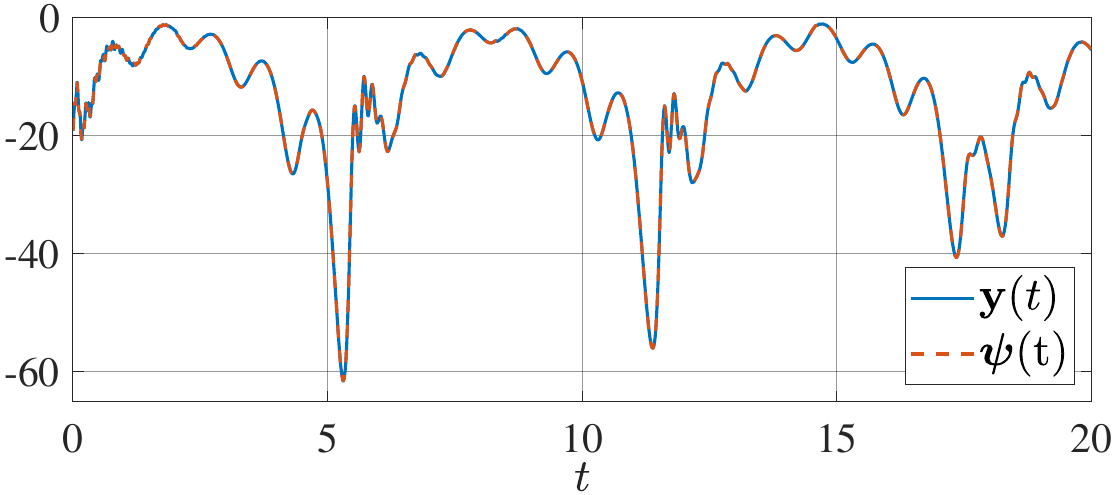}
    \caption{Time histories of outputs $\mathbf{y}$ (blue solid) and $\bm{\psi}$ (red dashed) of systems~$\Sigma$ and~$\Sigma^\prime$ in the interconnection in Fig.~\ref{fig:mRelateInterconnect}.}
    \label{fig:Example_mRelationOut}
\end{centering}
\end{figure}

\section{Conclusion}\label{sec:concl}
This paper has addressed the ASHC problem of nonlinear systems by providing methods for designing an abstract system and an interface that satisfy two key requirements: bounded output discrepancy and $m$-relation. Regarding the bounded output discrepancy requirement, we have proposed a PDE-based solution under the assumptions of asymptotic stabilisability of the concrete system and Lipschitz continuity of its output mapping (Theorem~\ref{thm:SimuDesign}). We have then shown that the $m$-relation condition can be satisfied by another PDE-based solution (Lemma~\ref{lem:mRelatePDE}). By combining the two proposed solutions, the overall solvability has been studied (Theorems~\ref{thm:mRelateDelta} and~\ref{thm:mCondi}), with the complete design procedure of the ASHC problem summarised in a flow chart (Fig.~\ref{fig:Flow_Chart}).
All the proposed results have been illustrated by a practical example. Note that while most results focus on input-affine concrete systems, the proposed solutions to the bounded output discrepancy requirement can be extended to general nonlinear systems, while the $m$-relation, in general, is not a core criterion in the ASHC field, as discussed at the end of Section~\ref{subsec:steps}.




\section*{References}
\bibliographystyle{IEEEtran}
\bibliography{IEEEabrv,mybib}

\begin{thebibliography}{10}
\providecommand{\url}[1]{#1}
\csname url@samestyle\endcsname
\providecommand{\newblock}{\relax}
\providecommand{\bibinfo}[2]{#2}
\providecommand{\BIBentrySTDinterwordspacing}{\spaceskip=0pt\relax}
\providecommand{\BIBentryALTinterwordstretchfactor}{4}
\providecommand{\BIBentryALTinterwordspacing}{\spaceskip=\fontdimen2\font plus
\BIBentryALTinterwordstretchfactor\fontdimen3\font minus
  \fontdimen4\font\relax}
\providecommand{\BIBforeignlanguage}[2]{{%
\expandafter\ifx\csname l@#1\endcsname\relax
\typeout{** WARNING: IEEEtran.bst: No hyphenation pattern has been}%
\typeout{** loaded for the language `#1'. Using the pattern for}%
\typeout{** the default language instead.}%
\else
\language=\csname l@#1\endcsname
\fi
#2}}
\providecommand{\BIBdecl}{\relax}
\BIBdecl

\bibitem{ref:antoulas2005approximation}
A.~C. Antoulas, \emph{Approximation of large-scale dynamical systems}.\hskip
  1em plus 0.5em minus 0.4em\relax Philadelphia, USA: SIAM, 2005.

\bibitem{ref:scarciotti2024interconnection}
G.~Scarciotti and A.~Astolfi, ``Interconnection-based model order reduction - a
  survey,'' \emph{Eur. J. Control}, vol.~75, p. 100929, 2024.

\bibitem{ref:girard2009hierarchical}
A.~Girard and G.~J. Pappas, ``Hierarchical control system design using
  approximate simulation,'' \emph{Automatica}, vol.~45, no.~2, pp. 566--571,
  2009.

\bibitem{ref:zamani2017compositional}
M.~Zamani and M.~Arcak, ``Compositional abstraction for networks of control
  systems: A dissipativity approach,'' \emph{IEEE Trans. Control Netw. Syst.},
  vol.~5, no.~3, pp. 1003--1015, 2017.

\bibitem{ref:rungger2016compositional}
M.~Rungger and M.~Zamani, ``Compositional construction of approximate
  abstractions of interconnected control systems,'' \emph{IEEE Trans. Control
  Netw. Syst.}, vol.~5, no.~1, pp. 116--127, 2016.

\bibitem{ref:smith2020approximate}
S.~W. Smith, M.~Arcak, and M.~Zamani, ``Approximate abstractions of control
  systems with an application to aggregation,'' \emph{Automatica}, vol. 119, p.
  109065, 2020.

\bibitem{ref:song2022robust}
Z.~Song, V.~Kurtz, S.~Welikala, P.~J. Antsaklis, and H.~Lin, ``Robust
  approximate simulation for hierarchical control of piecewise affine systems
  under bounded disturbances,'' in \emph{Proc. 41st Amer. Control Conf.}, 2022,
  pp. 1543--1548.

\bibitem{ref:zamani2016approximations}
M.~Zamani, M.~Rungger, and P.~M. Esfahani, ``Approximations of stochastic
  hybrid systems: A compositional approach,'' \emph{IEEE Transactions on
  Automatic Control}, vol.~62, no.~6, pp. 2838--2853, 2016.

\bibitem{ref:zhong2024hierarchical}
B.~Zhong, M.~Arcak, and M.~Zamani, ``Hierarchical control for cyber-physical
  systems via general approximate alternating simulation relations,''
  \emph{IFAC-PapersOnLine}, vol.~58, no.~11, pp. 13--18, 2024.

\bibitem{ref:fu2013hierarchical}
J.~Fu, S.~Shah, and H.~G. Tanner, ``Hierarchical control via approximate
  simulation and feedback linearization,'' in \emph{Proc. 32nd Am. Control
  Conf.}, 2013, pp. 1816--1821.

\bibitem{ref:tabuada2008approximate}
P.~Tabuada, ``An approximate simulation approach to symbolic control,''
  \emph{IEEE Trans. Autom. Control}, vol.~53, no.~6, pp. 1406--1418, 2008.

\bibitem{ref:girard2007approximation}
A.~Girard and G.~J. Pappas, ``Approximation metrics for discrete and continuous
  systems,'' \emph{IEEE Trans. Autom. Control}, vol.~52, no.~5, pp. 782--798,
  2007.

\bibitem{ref:schilders2008model}
W.~H. Schilders, H.~A. Van~der Vorst, and J.~Rommes, \emph{Model order
  reduction: theory, research aspects and applications}.\hskip 1em plus 0.5em
  minus 0.4em\relax Berlin, Heidelberg: Springer, 2008, vol.~13.

\bibitem{ref:niu2025Briging}
Z.~Niu, M.~F. Shakib, and G.~Scarciotti, ``Bridging abstraction-based
  hierarchical control and moment matching: A conceptual unification,'' in
  \emph{Proc. 64th IEEE Conf. Decis. Control}, 2025.

\bibitem{ref:obinata2012model}
G.~Obinata and B.~D. Anderson, \emph{Model reduction for control system
  design}.\hskip 1em plus 0.5em minus 0.4em\relax Springer Science \& Business
  Media, 2012.

\bibitem{ref:meijer2022finite}
T.~J. Meijer, S.~A.~N. Nouwens, V.~S. Dolk, B.~de~Jager, and W.~P. M.~H.
  Heemels, ``Finite-horizon minimal realizations for model predictive control
  of large-scale systems,'' in \emph{Proc. 61st IEEE Conf. Decis. Control},
  2022, pp. 1136--1141.

\bibitem{ref:jansen2017use}
J.~D. Jansen and L.~J. Durlofsky, ``Use of reduced-order models in well control
  optimization,'' \emph{Optim. Eng.}, vol.~18, no.~1, pp. 105--132, 2017.

\bibitem{ref:angeli1999forward}
D.~Angeli and E.~D. Sontag, ``Forward completeness, unboundedness
  observability, and their {Lyapunov} characterizations,'' \emph{Syst. Control
  Lett.}, vol.~38, no. 4-5, pp. 209--217, 1999.

\bibitem{ref:sontag2013mathematical}
E.~D. Sontag, \emph{Mathematical Control Theory: Deterministic Finite
  Dimensional Systems}.\hskip 1em plus 0.5em minus 0.4em\relax Springer Science
  \& Business Media, 2013, vol.~6.

\bibitem{ref:angeli2002lyapunov}
D.~Angeli, ``A {L}yapunov approach to incremental stability properties,''
  \emph{IEEE Trans. Autom. Control}, vol.~47, no.~3, pp. 410--421, 2002.

\bibitem{ref:milner1989communication}
R.~Milner, \emph{Communication and Concurrency}.\hskip 1em plus 0.5em minus
  0.4em\relax Prentice-Hall, Inc., 1989.

\bibitem{ref:pappas2000hierarchically}
G.~J. Pappas, G.~Lafferriere, and S.~Sastry, ``Hierarchically consistent
  control systems,'' \emph{IEEE Trans. Autom. Control}, vol.~45, no.~6, pp.
  1144--1160, 2000.

\bibitem{ref:tabuada2005hierarchical}
P.~Tabuada and G.~J. Pappas, ``Hierarchical trajectory refinement for a class
  of nonlinear systems,'' \emph{Automatica}, vol.~41, no.~4, pp. 701--708,
  2005.

\bibitem{ref:jouffroy2010tutorial}
J.~Jouffroy and T.~I. Fossen, ``A tutorial on incremental stability analysis
  using contraction theory,'' \emph{Model. Identif. Control}, vol.~31, no.~3,
  p.~93, 2010.

\bibitem{ref:samari2025model}
B.~Samari, A.~Nejati, and A.~Lavaei, ``Model order reduction from data with
  certification,'' in \emph{Proc. 64th IEEE Conf. Decis. Control}, 2025, pp.
  5800--5805.

\bibitem{ref:kalise2018polynomial}
D.~Kalise and K.~Kunisch, ``Polynomial approximation of high-dimensional
  {H}amilton--{J}acobi--{B}ellman equations and applications to feedback
  control of semilinear parabolic {PDE}s,'' \emph{SIAM J. Sci. Comput.},
  vol.~40, no.~2, pp. A629--A652, 2018.

\bibitem{ref:kalise2020robust}
D.~Kalise, S.~Kundu, and K.~Kunisch, ``Robust feedback control of nonlinear
  {PDE}s by numerical approximation of high-dimensional
  {H}amilton--{J}acobi--{I}saacs equations,'' \emph{SIAM J. Appl. Dyn. Syst.},
  vol.~19, no.~2, pp. 1496--1524, 2020.

\bibitem{ref:faedo2021approximation}
N.~Faedo, G.~Scarciotti, A.~Astolfi, and J.~V. Ringwood, ``On the approximation
  of moments for nonlinear systems,'' \emph{IEEE Trans. Autom. Control},
  vol.~66, no.~11, pp. 5538--5545, 2021.

\bibitem{ref:doebeli2024polynomial}
C.~Doebeli, A.~Astolfi, D.~Kalise, A.~Moreschini, G.~Scarciotti, and J.~Simard,
  ``A polynomial approximation scheme for nonlinear model reduction by moment
  matching,'' \emph{arXiv preprint arXiv:2412.13371}, 2024.

\bibitem{ref:astolfi2010model}
A.~Astolfi, ``Model reduction by moment matching for linear and nonlinear
  systems,'' \emph{IEEE Trans. Autom. Control}, vol.~55, no.~10, pp.
  2321--2336, 2010.

\bibitem{ref:scarciotti2017nonlinear}
G.~Scarciotti and A.~Astolfi, ``Nonlinear model reduction by moment matching,''
  \emph{Found. Trends Syst. Control}, vol.~4, no. 3-4, pp. 224--409, 2017.

\bibitem{ref:lee2013introduction}
J.~M. Lee, \emph{Introduction to Smooth Manifolds}, 2nd~ed.\hskip 1em plus
  0.5em minus 0.4em\relax New York, NY, USA: Springer, 2012.

\bibitem{ref:girard2007approximate}
A.~Girard and G.~J. Pappas, ``Approximate bisimulation relations for
  constrained linear systems,'' \emph{Automatica}, vol.~43, no.~8, pp.
  1307--1317, 2007.

\bibitem{ref:mosek}
\BIBentryALTinterwordspacing
{MOSEK ApS}, \emph{The {MOSEK} {API} for {MATLAB} 11.0.28}, 2019, {A}ccessed:
  Mar. 10, 2026. [Online]. Available:
  \url{https://docs.mosek.com/latest/matlabapi/index.html}
\BIBentrySTDinterwordspacing

\bibitem{ref:YALMIP}
J.~L{\"{o}}fberg, ``{YALMIP}: A toolbox for modeling and optimization in
  {MATLAB},'' in \emph{Proc. 13th IEEE Conf. Comput.-Aided Control Syst. Des.},
  2004.

\end{thebibliography}

\begin{IEEEbiography}[{\includegraphics[width=1in,height=1.25in,clip,keepaspectratio]{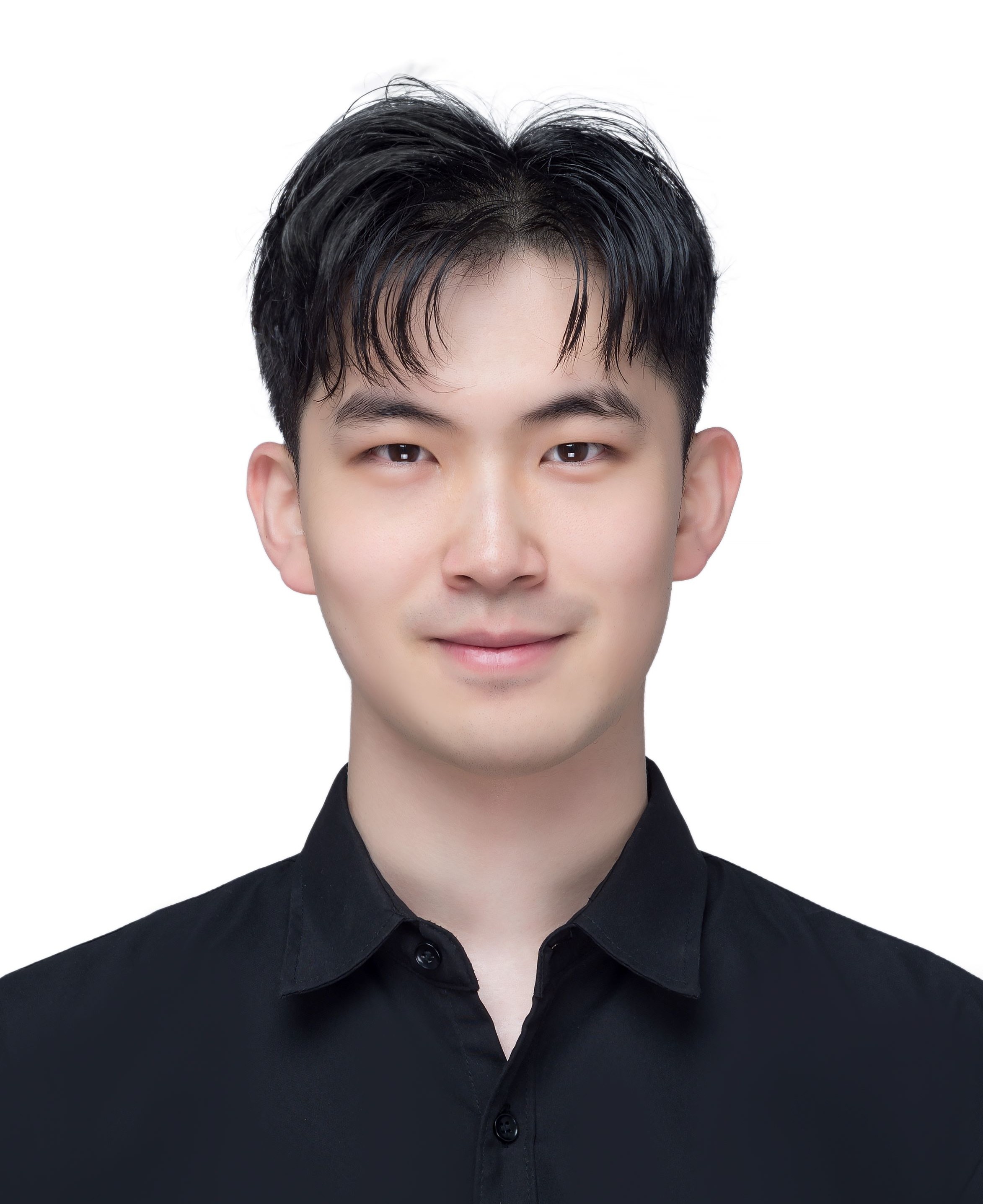}}]{Zirui Niu} (Graduate Student Member, IEEE)
was born in Shandong, China. He received the B.Eng. (Hons) degree in electrical and electronic engineering from the University of Liverpool, UK, in 2020, and the M.Sc. degree in control systems from Imperial College London, UK, in 2021. Since 2022, he has been pursuing a Ph.D. in control theory at Imperial College London, UK, supported by the CSC-Imperial Scholarship. His research interests include output regulation, hybrid systems, model reduction, approximate simulation of complex systems, and data-driven control. He was the recipient of the MSc Control Systems Outstanding Achievement Prize (2021) and the Hertha Ayrton Centenary Prize (2021) for outstanding performance in MSc Control Systems and the outstanding master’s thesis in the Electrical and Electronic Engineering Department at Imperial College London, respectively.
\end{IEEEbiography}

\begin{IEEEbiography}
[{\includegraphics[width=1in,height=1.25in,clip,keepaspectratio]{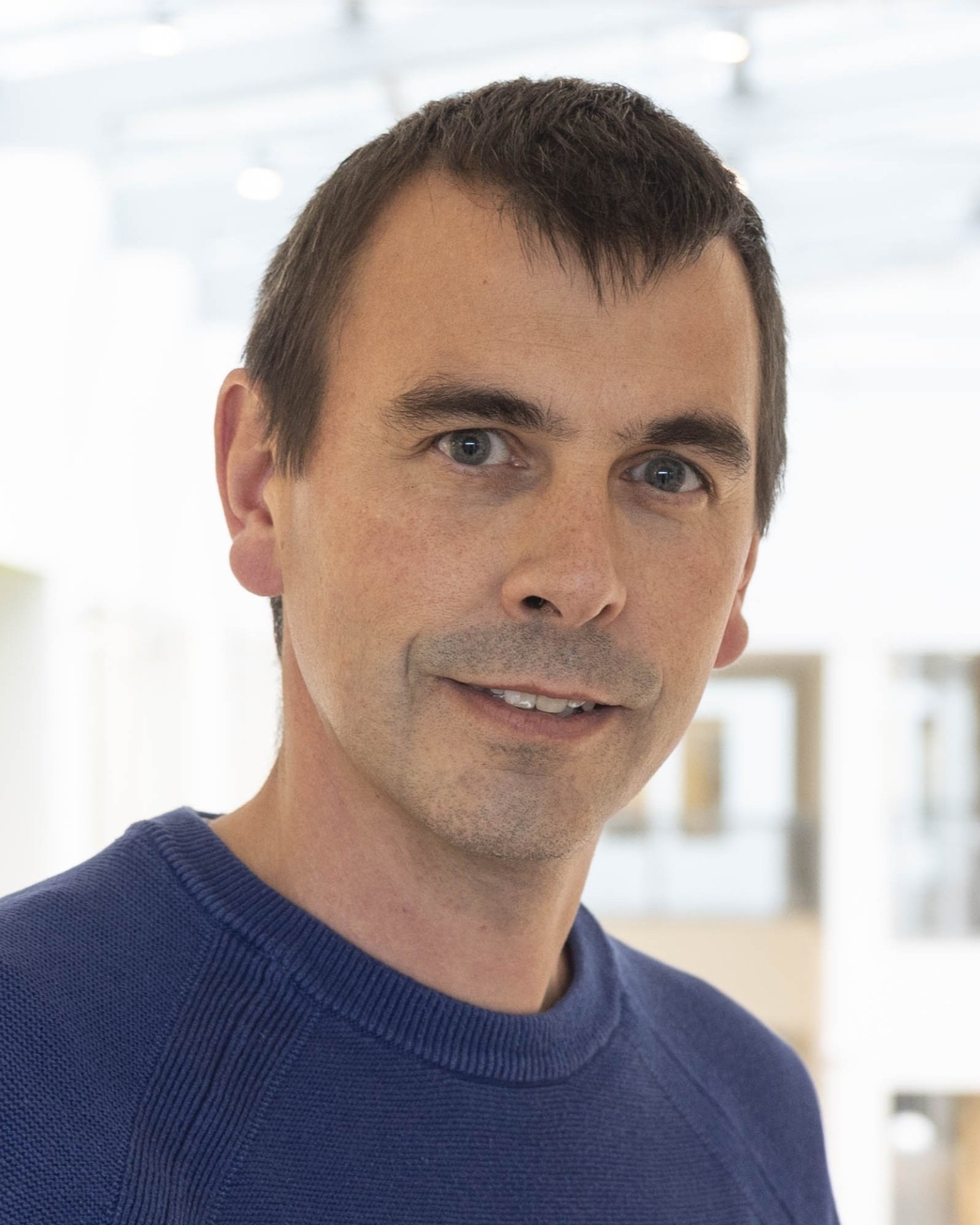}}]
{Antoine Girard} (Fellow, IEEE) is a Senior Researcher at CNRS and a member of the Laboratory of Signals and Systems. He is also an Adjunct Professor at CentraleSupélec, Université Paris-Saclay. He received the Ph.D. degree from Grenoble Institute of Technology, in 2004. From 2004 to 2006, he held postdoctoral positions at University of Pennsylvania and Université Grenoble-Alpes. From 2006 to 2015, he was an Assistant/Associate Professor at the Université Grenoble-Alpes. His main research interests deal with analysis and control of hybrid systems with an emphasis on computational approaches, formal methods and applications to cyber-physical and autonomous systems. Antoine Girard is an IEEE Fellow. In 2015, he was appointed as a junior member of the Institut Universitaire de France (IUF). In 2016, he was awarded an ERC Consolidator Grant. He received the George S. Axelby Outstanding Paper Award from the IEEE Control Systems Society in 2009, the CNRS Bronze Medal in 2014, and the European Control Award in 2018.
\end{IEEEbiography}

\begin{IEEEbiography}[{\includegraphics[width=1in,height=1.25in,clip,keepaspectratio]{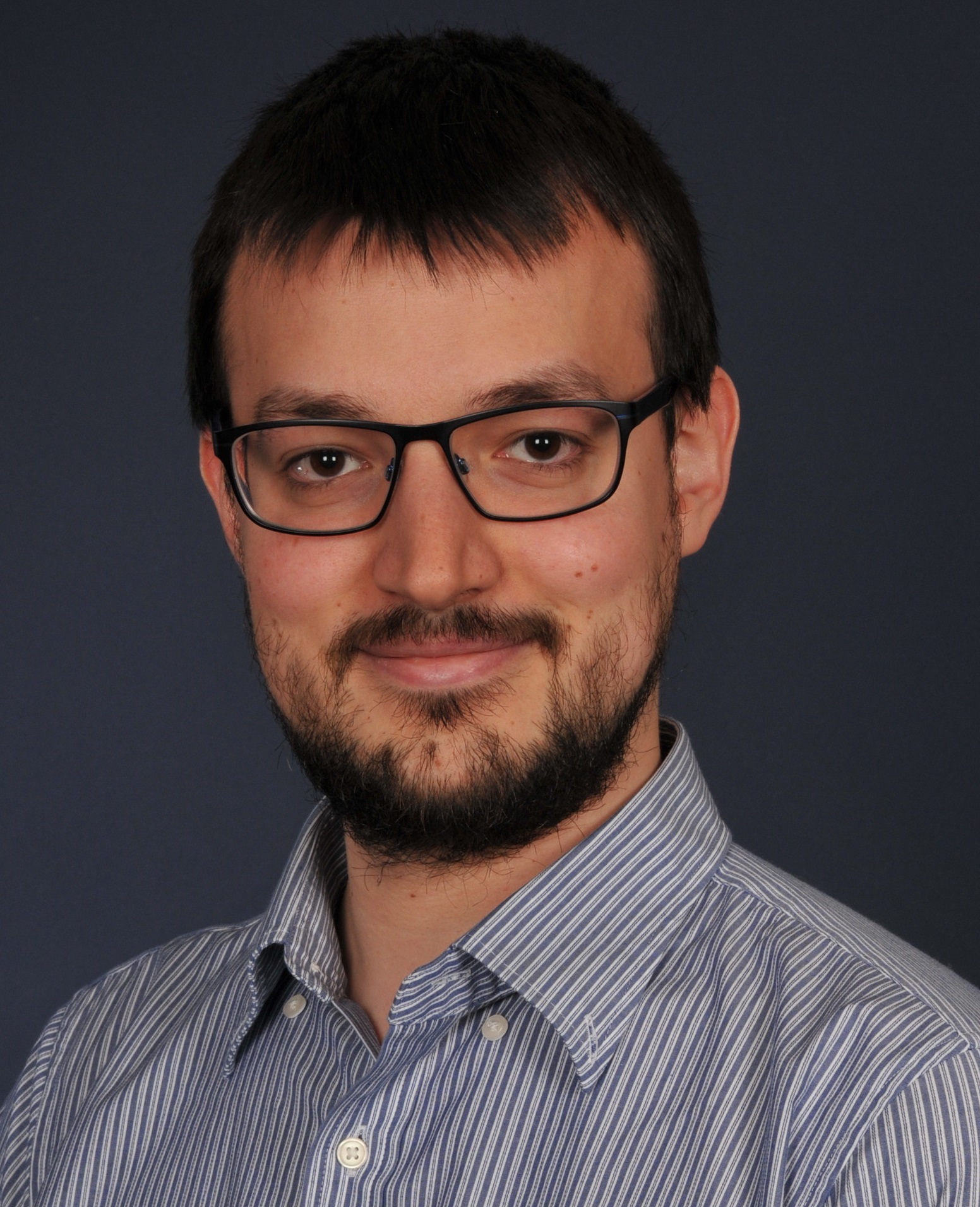}}] {Giordano Scarciotti} (Senior Member, IEEE) received his B.Sc. and M.Sc. degrees in Automation Engineering from the University of Rome ``Tor Vergata'', Italy, in 2010 and 2012, respectively. In 2012 he joined the Control and Power Group, Imperial College London, UK, where he obtained a Ph.D. degree in 2016. He also received an M.Sc. in Applied Mathematics from Imperial in 2020. He is currently an Associate Professor in Control Theory at Imperial. He was a visiting scholar at New York University in 2015, at University of California Santa Barbara in 2016, and a Visiting Fellow of Shanghai University in 2021-2022. He is the recipient of an Imperial College Junior Research Fellowship (2016), of the IET Control \& Automation PhD Award (2016), the Eryl Cadwaladr Davies Prize (2017), an ItalyMadeMe award (2017) and the IEEE Transactions on Control Systems Technology Outstanding Paper Award (2023). He is a member of the EUCA Conference Editorial Board, of the IFAC and IEEE CSS Technical Committees on Nonlinear Control Systems and has served in the International Programme Committees of multiple conferences. He is Associate Editor of Automatica and Guest Associate Editor of Nonlinear Analysis: Hybrid Systems. He was the National Organising Committee Chair for the EUCA European Control Conference (ECC) 2022, and of the 7th IFAC Conference on Analysis and Control of Nonlinear Dynamics and Chaos 2024, and the Invited Session Chair and Editor for the IFAC Symposium on Nonlinear Control Systems 2022 and 2025, respectively. He is the General Co-Chair of ECC 2029. 
\end{IEEEbiography}

\end{document}